\documentclass[11pt]{article}
\pdfoutput=1

\usepackage[T1]{fontenc}
\usepackage[utf8]{inputenc}
\usepackage{lmodern}
\usepackage{microtype}
\usepackage[margin=1in]{geometry}
\usepackage{amsmath,amssymb,amsthm,mathtools}
\usepackage{booktabs,tabularx,array,multirow}
\usepackage{float}
\usepackage[section]{placeins}
\usepackage{enumitem}
\usepackage{xcolor}
\usepackage{graphicx}
\usepackage{tikz}
\usetikzlibrary{arrows.meta,positioning,fit}
\usepackage{hyperref}
\usepackage[nameinlink,noabbrev]{cleveref}
\usepackage[style=authoryear,maxcitenames=3,maxbibnames=99,giveninits=true,uniquename=init,doi=false,isbn=false,url=true,eprint=true]{biblatex}
\AtBeginBibliography{\small\sloppy}
\hypersetup{
  colorlinks=true,
  linkcolor=blue!45!black,
  citecolor=blue!45!black,
  urlcolor=blue!55!black,
  pdftitle={Beyond Self-Resolution: Settlement Factorization for Robust Natural Language Mechanisms},
  pdfauthor={Nicolas Della Penna}
}

\newtheorem{definition}{Definition}
\newtheorem{proposition}{Proposition}
\newtheorem{theorem}{Theorem}
\newtheorem{corollary}{Corollary}
\newtheorem{lemma}{Lemma}
\newtheorem{conjecture}{Conjecture}
\newtheorem{remark}{Remark}
\newtheorem{assumption}{Assumption}

\newcommand{\E}{\mathbb{E}}
\newcommand{\Prb}{\mathbb{P}}
\newcommand{\TV}{\operatorname{TV}}
\newcommand{\Law}{\mathcal{L}}
\newcommand{\calR}{\mathcal{R}}
\newcommand{\calE}{\mathcal{E}}
\newcommand{\calY}{\mathcal{Y}}
\newcommand{\calZ}{\mathcal{Z}}

\newcommand{\ind}{\mathrel{\perp\!\!\!\perp}}

\title{\textbf{Beyond Self-Resolution: Settlement Factorization for Robust Natural Language Mechanisms}}
\author{Nicolas Della Penna\\\href{mailto:nikete@grouplang.ai}{nikete@grouplang.ai}}
\date{July 2026}

\begin{document}
\maketitle

\begin{abstract}
Language models increasingly mediate \emph{paid advice}: agents submit open-ended forecasts, recommendations, plans, and evidence; a principal acts on the reports; and the mechanism later pays the contributors. Advice should influence the public decision, but no adviser should write the answer key used to evaluate it. We formalize the separation as \emph{settlement factorization}: reports are hardened into official records, a public decision record $Z$ may use all advice, and each paid adviser is scored against a label whose production is externalized from their own report, conditional on $Z$. The central result is an analogue of the revelation principle for this setting: resampling the paid report from a committed \emph{ghost} distribution inside the settlement channel equips every mechanism with an influence-free reference within a factor of two of the best achievable, so factorization is a normal form, and own-report leakage $\varepsilon$---measured in total variation---is an intrinsic invariant of any mechanism, identified behaviorally by worst-case incentive erosion. The invariant has an exact price: with payment kernels of label sensitivity $L$, truthful margins degrade by at most $2L\varepsilon$, and the constant is tight---half own-label manipulation, half pandering to a biased evaluator---so faithful advice survives outside decision interests $D$ whenever the externalized margin satisfies $\gamma>D+2L\varepsilon$. Settling a growing crowd on the shared decision record dilutes every margin exponentially in the informational redundancy, while a factorized leave-one-out label sustains a constant margin at unit stakes. Stakes outbid decision interests but never evaluator capture; randomized reference settlement makes integrity a threshold good priced per unit of total variation; and differential privacy of the evaluator in the paid report certifies $\varepsilon$ by construction.
\end{abstract}

\section{Introduction}

A model recommends which vendor to hire. A research agent proposes an experiment. A reviewer writes an assessment. A provider on an agent marketplace describes its skills, suggests a plan, and predicts whether the plan will work. In each case, a principal wants to reward useful language without delegating the entire institutional decision to the report itself.

It is helpful to call these systems \emph{advice mechanisms}. Their messages are richer than scalar forecasts. An adviser may provide factual claims, conditional predictions, a recommended action, a decomposition of the task, caveats, and supporting evidence in one report. The principal retains the right to combine the reports, apply a constitution, and choose an action. This is neither ordinary preference revelation nor mere text evaluation: the report is intended to change what the principal does.

That intended influence creates a design boundary:

\begin{quote}
\emph{The report may be evidence for the decision, but it must not write its own answer key.}
\end{quote}

The first clause is why a leave-one-out decision rule is often inappropriate. Good advice should affect the public decision. The second clause is why a single unconstrained ``judge the report and decide its reward'' call is fragile. If adviser $i$ can steer the label against which adviser $i$ is scored, a high payment may reflect evaluator control rather than information, effort, or decision value.

\textcite{dellapenna2024nlm} isolates a positive regime in which a foundation model can interpret rich reports and directly implement allocations and transfers under strong accuracy and information-overdetermination conditions.\footnote{That paper measures pipeline quality by a multiplicative welfare guarantee (a $\delta$-sufficient world model). The total-variation leakage used here composes with that notion: by \cref{lem:tv-expectation}, a total-variation bound of $\varepsilon$ on any label or decision kernel translates into an additive payoff perturbation of at most $L\varepsilon$ for every $L$-bounded objective, so the two quality measures can be used in the same accounting.} The present paper does not dispute that result. It asks what an institution should expose once it needs a robust, inspectable, and modular settlement process.

The answer is \emph{settlement factorization}. Reports are hardened into official advice records. A public decision record may use all admissible advice. For each paid adviser, however, the mechanism also constructs an evaluation label whose production is externalized from that adviser's own report, conditional on the public decision and the institution's authorized data. A written resolution is then generated from the settled decision and explicit provenance rather than serving as an independent source of payoff-relevant truth.

The paper's central result says this is not merely a recommended architecture. In the spirit of the revelation principle, every advice mechanism---whatever its internal structure---already admits a factorized representation: resampling the paid report from a committed \emph{ghost} distribution inside its own settlement channel manufactures an influence-free reference, provably within a factor of two of the best possible, so that designing within the factorized class is without loss of generality and the residual own-report leakage $\varepsilon_i$ is a well-defined invariant of any mechanism. The rest of the paper prices that invariant exactly, shows worst-case incentive erosion identifies it, and shows how to purchase, certify, and measure it.

This architecture is deliberately agnostic about how the ideal evaluation label is obtained. It may be a directly verified outcome, a holdout audit, a proxy, a peer or reference report, a self-resolving market label, an evidence-weighted judgment, a vector of potential outcomes, or a hybrid of these. The paper is therefore about going \emph{beyond verification}: direct ground truth is one clean special case, not a prerequisite for settlement design.

The distinction also clarifies the relationship to a rapidly developing literature. Textual scoring mechanisms ask how to score strategic text against ground truth or peer text \parencite{wu2024elicitationgpt,lu2024informative,lu2025aligned}. Rationale and evidence mechanisms ask how to elicit explanations or evidence and how those objects improve aggregation \parencite{srinivasan2025tell,hossain2026evidence}. Rich-query aggregation asks what additional structured questions recover information that scalar forecasts omit \parencite{frongillo2025queries}. Self-resolving markets ask how later or better-informed reports can proxy for unverifiable outcomes \parencite{srinivasan2023selfresolving}. Decision markets, performative prediction, compensation rules, proxy forecasting, and joint scoring rules ask how to preserve forecast incentives when reports affect actions or outcomes \parencite{boutilier2011experts,wang2021decision,wang2023proxy,oesterheld2023performative,hudson2024joint}. These mechanisms can create the ideal truthful-advice margin used in our analysis. Settlement factorization studies the separate question of whether the operational pipeline preserves that margin or lets the paid report contaminate its own evaluator.

\paragraph{Contributions.}
The paper makes seven contributions.

\begin{enumerate}[leftmargin=1.5em]
    \item It reframes natural language mechanisms as mechanisms for \emph{paid advice}. The principal's decision and the adviser's evaluation are different institutional objects. Advice may influence the first; the second must be externally anchored.
    \item It defines settlement factorization with four explicit objects: hardened advice, a public decision record, agent-specific externalized labels, and provenance-linked writing. Conditional factorization permits the public decision to depend on all reports while excluding direct own-report influence from the label used to pay the same adviser.
    \item \textbf{The central result: factorization is a normal form.} A ghost-reference construction makes the factorized class without loss of generality, and a converse shows worst-case margin erosion over a canonical probe class \emph{identifies} the leakage invariant exactly---an operational meaning independent of pipeline internals. In the canonical redundant-information family of self-resolution, \emph{any} payment settled on the shared record has margins vanishing exponentially in redundancy, while a factorized leave-one-out label keeps a non-vanishing margin at unit stakes---the paper's title as a theorem.
    \item It gives a modular incentive result. Any external scoring or decision module that supplies a truthful-advice margin continues to work under exact factorization. Proper scoring yields a score-divergence margin; under log scoring this equals conditional mutual information.
    \item It proves an approximate robustness theorem and shows it is exactly tight. If own-label leakage is $\varepsilon_i$ and label sensitivity is $L_i$, truthful margins lose at most $2L_i\varepsilon_i$; a matching construction realizes the full budget, with own-report manipulation and pandering to a biased evaluator each accounting for $L_i\varepsilon_i$. With bounded direct interest $D_i$ in the principal's decision, a sufficient condition is $\gamma_i>D_i+2L_i\varepsilon_i$. A composition lemma bounds the leakage of hybrid labels built from conditionally independent components; a counterexample shows the independence requirement cannot be dropped.
    \item It derives three design levers from the same primitives. The faithfulness condition is scale-invariant, so stakes can outbid any bounded decision interest but never evaluator capture, making the normalized margin $\underline m_i/L_i>2\varepsilon_i$ a hard feasibility frontier. Randomized reference settlement reduces effective leakage linearly in the audit rate. And because black-box audits cannot certify small leakage over rich label spaces, certification must be architectural: differential privacy of the evaluator in the paid report bounds $\varepsilon_i$ by construction, with composition guarantees that survive adaptively chosen label components. Together these give a principal's problem in which integrity is a threshold good with an explicit price per unit of total variation, and the audit rate $p_i^\ast$ is its solution.
    \item It turns the theorem primitives into an empirical agenda: estimate advice value $\gamma_i$, own-label leakage $\varepsilon_i$, decision interest $D_i$, constitution sensitivity, and memo fidelity rather than relying only on opaque end-to-end accuracy.
\end{enumerate}

The claim is intentionally narrower than ``factorization solves truthful advice.'' It does not. Decision coupling, peer correlation, collusion, identity, fairness, and evaluator quality remain substantive mechanism-design problems. Factorization is the architecture that keeps solutions to those problems from being undone by evaluator capture.

\paragraph{Roadmap.}
\Cref{sec:literature,sec:example,sec:model} fix the objects: paid advice, a marketplace example, and the four settlement objects with their authorized information flows. \Cref{sec:exact} shows exact factorization preserves any upstream truthful margin. \Cref{sec:leakage} is the core: the $2L_i\varepsilon_i$ price and its exact tightness, the behavioral converse, the ghost-reference normal form, the dilution family that quantifies the title, and the levers that turn the price schedule into design. \Cref{sec:labels,sec:certified} construct externalized labels without literal ground truth and certify their leakage architecturally; \cref{sec:language,sec:limits,sec:measurement} say what is language-specific, what factorization does not solve, and how to estimate every primitive; and \cref{sec:conjecture} states the mechanism-level normal form that remains conjectural.

\section{Paid advice and the neighboring literature}
\label{sec:literature}

\subsection{Advice rather than delegated control}

An advice mechanism has three conceptually distinct stages. First, advisers communicate information and recommendations. Second, a principal chooses what to do. Third, the institution evaluates and pays the advice. Collapsing these stages is attractive because a language model can execute all three in one prompt. It is also precisely what makes causal responsibility hard to audit.

The advice framing differs from delegated control. The principal may reject an adviser's recommendation, combine several reports, enforce constraints, or choose a stochastic action for incentive reasons. It also differs from ordinary forecasting. Reports may contain conditional predictions, causal claims, action proposals, and evidence whose value depends on the principal's objective. A faithful report is therefore not merely a calibrated scalar probability. It is a report whose hardened content faithfully represents the adviser's information and whose recommendation is optimal relative to the stated constitution.

Two influence channels must be separated:
\begin{align*}
\text{advice} &\longrightarrow \text{public decision}, &&\text{intended};\\
\text{advice}_i &\longrightarrow \text{label used to pay }i, &&\text{dangerous unless explicitly incentive-compatible}.
\end{align*}

The first channel creates the value of advice. The second can make a report self-fulfilling, performative, or simply self-grading. Settlement factorization does not suppress the first channel. It exposes and bounds the second.

\subsection{What the adjacent mechanisms contribute}

\Cref{tab:literature-map} summarizes the closest lines of work. The rows are not mutually exclusive. A deployed mechanism may combine several of them.

\begin{table}[H]
\centering
\small
\begin{tabularx}{\textwidth}{>{\raggedright\arraybackslash}p{0.19\textwidth}>{\raggedright\arraybackslash}p{0.25\textwidth}>{\raggedright\arraybackslash}X>{\raggedright\arraybackslash}p{0.23\textwidth}}
\toprule
\textbf{Line of work} & \textbf{What is elicited} & \textbf{Where the incentive comes from} & \textbf{Relation to this paper}\\
\midrule
NLM self-resolution \parencite{dellapenna2024nlm} & Unrestricted natural-language reports & A capable world model interprets reports and jointly selects outcomes and transfers & Positive end-to-end baseline; this paper externalizes settlement objects.\\
\addlinespace
Text scoring and textual peer prediction \parencite{wu2024elicitationgpt,lu2024informative,lu2025aligned} & Text evaluated against ground-truth or peer text & Proper numerical reductions, predictive models, or aligned score construction & Candidate payment kernels and truthful margins; does not by itself guarantee label provenance.\\
\addlinespace
Rationales and evidence \parencite{srinivasan2025tell,hossain2026evidence} & Beliefs together with explanations or supporting evidence & Deliberation, evidence-sensitive liquidity, staking, or endogenous resolution & Candidate hardening and label-production modules; factorization asks whether evidence can steer its submitter's own label.\\
\addlinespace
Rich-query aggregation \parencite{frongillo2025queries} & Structured answers beyond first-order forecasts & Truthful query mechanisms with accuracy-complexity guarantees & Alternative to unrestricted prose; can implement externalized label construction.\\
\addlinespace
Self-resolving and peer mechanisms \parencite{miller2005peer,srinivasan2023selfresolving} & Reports about unverifiable states & Better-informed later reports or correlated peers proxy for truth & Important beyond-verification label sources; vulnerable if reference selection is report-controlled.\\
\addlinespace
Decision and performative scoring \parencite{boutilier2011experts,wang2023proxy,oesterheld2023performative,hudson2024joint} & Conditional predictions used to choose actions & Compensation, proxies, stochastic choice, fixed points, or zero-sum joint scores & Supplies an externalized margin when advice changes actions or outcomes. Factorization protects the resulting score channel.\\
\addlinespace
Token-level generation auctions \parencite{duetting2024mechanism,liu2026ads} & Scalar bids over next-token distributions during joint generation & Monotone distribution aggregation with second-price payments & Governs how joint content is \emph{produced}; orthogonal to how a finished report is scored and settled.\\
\addlinespace
Settlement factorization (this paper) & Hardened advice, public decision, externalized label, and resolution & Any valid upstream scoring or decision mechanism, minus a quantified leakage tax & Architectural normal form separating useful decision influence from own-evaluator influence.\\
\bottomrule
\end{tabularx}
\caption{The paper is complementary to mechanisms that construct truthful scores. Those mechanisms aim to create an ideal advice margin; settlement factorization asks whether the operational institution preserves it.}
\label{tab:literature-map}
\end{table}

This division of labor matters for novelty. Proper textual scoring is a scoring primitive, not the contribution claimed here. Evidence markets provide a structured market with evidence-sensitive liquidity, bounded loss, endogenous resolution, and approximate incentive guarantees \parencite{hossain2026evidence}; they are substantially more operational than the abstract factorization studied here. Joint scoring rules address a distinct decision-coupling failure: an accuracy-seeking adviser can manipulate a principal toward predictable actions, and zero-sum competition can remove that influence incentive under the paper's conditions \parencite{hudson2024joint}. Our theorem begins after such a mechanism has generated an externalized truthful margin and asks how much implementation leakage that margin can tolerate.

The oldest neighbor is agency theory. The informativeness principle says compensation should condition on signals informative about the agent's hidden action \parencite{holmstrom1979moral}; the paid report is that principle's adversarial inversion---a signal maximally \emph{controllable} by the agent---and $\varepsilon_i$ measures its residual presence in the performance measure. The literature on distorted and gameable performance measures \parencite{baker1992incentive} and on multitask crowding-out \parencite{holmstrom1991multitask} predicts that agents optimize the measure rather than the objective; \cref{thm:leakage,prop:tightness} sharpen that prediction for this setting into an exactly tight agency tax of $2L_i\varepsilon_i$, and factorization is the informativeness principle applied to strategic contamination rather than to noise.

\subsection{Position in the current LLM mechanism-design stack}

A contemporaneous workshop program on game theory and mechanism design with large language models makes the stack visible \parencite{llmincentives2026}. Textual skill allocation addresses who should receive a task; communication-cost models qualify when richer prompts are worth their burden \parencite{dong2026rightsizing}; evidence and scoring papers address what reports should contain and how to score them; in-context credit assignment addresses how to divide value among contributors; and work on clone robustness, mecha-nudges, fairness, and collusion addresses attacks and constraints around the market \parencite{harris2026core,hays2026candidacy,frey2026mecha,hosseini2025fairness,carissimo2025orchestration,luo2026collusion,galanis2026aggregation,kang2026reasoning}.

Settlement factorization occupies the interface between these layers. For example, a textual-skill allocation rule may select an agent, an evidence market may collect support for the output, an aligned textual score may value the report, and a least-core rule may divide surplus. None of those choices alone determines whether the selected provider can contaminate the label used to pay itself. Conversely, factorization does not solve clone entry, coalition payments, distributive fairness, or collusion.

Several accepted workshop manuscripts were not publicly available at the time of writing, including work on manipulation-resistant AI oracles, textual-skill task allocation, probabilistic coherence, endogenous reliance, wagering-based aggregation, and token-level advertising \parencite{monroe2026oracles,vonallmen2026skills,chadwick2026coherence,jiang2026reliance,luo2026wagering,liu2026ads}. We therefore use their public titles only to locate the present paper in the emerging research stack and make no technical priority claim about their unpublished results.

A related but orthogonal line designs mechanisms \emph{inside} the generation process itself: \textcite{duetting2024mechanism} auction next-token distributions among bidding LLM agents, showing that two natural incentive properties are equivalent to a monotonicity condition on distribution aggregation, which in turn enables a second-price payment rule; token-level advertising \parencite{liu2026ads} continues this direction. Those mechanisms govern how joint content is produced. Settlement factorization governs how a finished report is scored and settled, and the two layers compose.

\subsection{The systems-security analogue: information-flow control}
\label{sec:ifc}

Settlement factorization is the incentive-theoretic counterpart of privilege separation in LLM-agent security. Prompt-injection defenses such as the dual-LLM pattern and CaMeL extract control and data flows from the trusted query so that untrusted content, processed by a quarantined model, can never alter control flow \parencite{debenedetti2025camel}; design-pattern catalogues generalize the approach to families of agent architectures \parencite{beurer2025design}. The hardening map $h_i^C$ plays the role of the quarantine boundary, and the taint being tracked is not attacker instructions but an adviser's strategic influence on its own settlement. Two differences matter. First, the advice setting cannot simply block the tainted channel: the report is \emph{supposed} to influence the public decision, so factorization must be conditional rather than absolute. Second, rather than certifying a binary property, \cref{thm:leakage} prices residual taint as a quantitative budget $2L_i\varepsilon_i$. The empirical premise that model-mediated evaluators are influenceable is well documented: LLM judges exhibit self-preference causally linked to recognizing their own generations \parencite{panickssery2024llm}, one concrete channel through which $\varepsilon_i>0$ arises even across nominally separate module calls.

\section{Running example: advice in an agent marketplace}
\label{sec:example}

A requester wants to migrate a production system. The constitution says that expected completion time, reliability, security risk, reversibility, and cost all matter. Providers submit free-form reports containing skill claims, a plan, conditional forecasts, references, caveats, and proposed milestones.

A one-call implementation could ask a language model to read the reports, choose a provider, decide whether the work succeeded, explain the result, and assign payments. The same model context would then perform four institutionally different acts: interpret advice, choose an action, produce the evaluator, and write the public account.

A factorized implementation instead creates the following objects.

\begin{enumerate}[leftmargin=1.5em]
    \item \textbf{Hardened advice.} Each report becomes an admissible record: normalized claims, cited spans, declared assumptions, predicted milestone distributions, dependency lists, and provenance. Unsupported persuasive language remains visible to humans but is not silently promoted into official evidence.
    \item \textbf{Public decision record.} The institution selects a provider and plan and records clause-level findings, such as expected reliability, budget, and migration risk. This object may use every provider's admissible advice, including the selected provider's report.
    \item \textbf{Agent-specific evaluation label.} Payment to provider $i$ is anchored to authorized evidence external to $i$'s own report: execution traces, hidden tests, requester acceptance, later audits, cross-fitted judgments from other reports, or a precommitted hybrid. The selected action may determine which outcomes are observed, but provider $i$ cannot directly rewrite how those observations are converted into its own label.
    \item \textbf{Written resolution.} A human-readable memo explains the public decision and later settlement using the decision record and provenance. The memo does not independently choose a different payoff-relevant truth.
\end{enumerate}

The allowed and disallowed paths are shown in \cref{fig:factorization}.

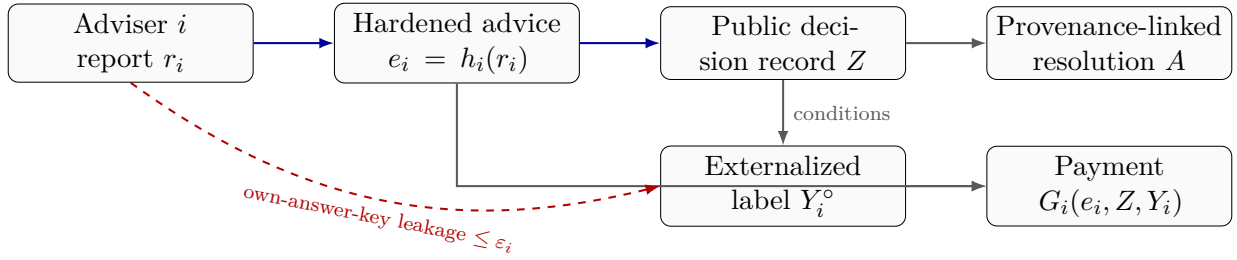
\begin{figure}[H]
\centering
\resizebox{0.98\textwidth}{!}{%
\begin{tikzpicture}[
    node distance=9mm and 11mm,
    box/.style={draw, rounded corners, align=center, minimum height=10mm, text width=31mm, fill=gray!4},
    good/.style={-{Latex[length=2mm]}, thick, blue!60!black},
    neutral/.style={-{Latex[length=2mm]}, thick, black!65},
    bad/.style={-{Latex[length=2mm]}, thick, red!70!black, dashed}
]
\node[box] (report) {Adviser $i$ report $r_i$};
\node[box, right=of report] (hard) {Hardened advice $e_i=h_i(r_i)$};
\node[box, right=of hard] (public) {Public decision record $Z$};
\node[box, below=of public] (label) {Externalized label $Y_i^\circ$};
\node[box, right=of public] (memo) {Provenance-linked resolution $A$};
\node[box, right=of label] (pay) {Payment $G_i(e_i,Z,Y_i)$};

\draw[good] (report) -- (hard);
\draw[good] (hard) -- (public);
\draw[neutral] (public) -- (memo);
\draw[neutral] (public) -- node[right, font=\scriptsize]{conditions} (label);
\draw[neutral] (label) -- (pay);
\draw[neutral] (hard) |- (pay);
\draw[bad, bend right=25] (report.south) to node[below, sloped, font=\scriptsize]{own-answer-key leakage $\le\varepsilon_i$} (label.west);
\end{tikzpicture}%
}
\caption{Advice should influence the public decision. Settlement factorization blocks or bounds the separate path from an adviser's report into the label used to pay that adviser: the operational label $Y_i$ may deviate from the reference $Y_i^\circ$ by at most $\varepsilon_i$ in conditional total variation (\cref{def:leaky}), and \cref{thm:leakage} prices this budget at $2L_i\varepsilon_i$.}
\label{fig:factorization}
\end{figure}

This example also separates factorization from direct verification. Hidden tests are one possible label source. When they are incomplete, the institution may combine tests with audits, peer evidence, market forecasts, or model-mediated judgments. The design question is not ``is the label perfect?'' but ``what is the label, who can influence it, and how does its imperfection affect incentives?''

\section{Model}
\label{sec:model}

\subsection{Environment and advice}

There are $n$ advisers. A latent environment $\theta\in\Theta$ is drawn from a common prior. Adviser $i$ observes a private signal $s_i\in S_i$; write $s=(s_1,\ldots,s_n)$. The institution publishes a natural-language constitution $C$ describing admissibility, the principal's objective, available actions, and payment rules.

Adviser $i$ submits a report $r_i\in\calR_i$. Reports may contain prose, structured fields, files, citations, or tool traces. A report hardener
\[
    h_i^C:\calR_i\to\calE_i
\]
maps the report into an official advice record $e_i=h_i^C(r_i)$. The superscript emphasizes that hardening is constitution-dependent: the same sentence may be admissible under one rulebook and irrelevant under another.

Let
\[
    \phi_i^C:S_i\to\calE_i
\]
be the truthful hardened-advice encoder. It represents the official content that a faithful report should induce when adviser $i$ observes $s_i$.

\begin{definition}[Faithful report set]
For signal $s_i$, define
\[
    T_i^C(s_i)=\{r_i\in\calR_i:h_i^C(r_i)=\phi_i^C(s_i)\}.
\]
A reporting policy is faithful if its support is contained in $T_i^C(s_i)$ for every signal.
\end{definition}

The definition is intentionally semantic rather than syntactic. Faithful paraphrases are strategically equivalent. A constitution may also require a recommendation or plan to satisfy a decision criterion. In that case $\phi_i^C(s_i)$ includes the Bayes-optimal recommendation under the adviser's posterior and the stated principal objective. Misstating evidence or knowingly recommending a dominated action changes the hardened record.

Write $e(r)=(h_1^C(r_1),\ldots,h_n^C(r_n))$ and $e_{-i}$ for the profile excluding adviser $i$.

\subsection{Public decisions, externalized labels, and writing}

The mechanism receives authorized outside data $X$ and randomness $U$. It produces a public decision record
\[
    Z=\rho_C(e(r),X,U)\in\calZ.
\]
The record may include a chosen action $a=d(Z)$, rankings, allocations, clause-level findings, and public commitments. It may depend on all advice.

For each adviser $i$, define a reference evaluation label
\[
    Y_i^\circ=\rho_{i,C}^\circ(Z,e_{-i},X_i,U_i)\in\calY_i.
\]
The label may depend on the public decision, other advisers' evidence, authorized holdout data, and mechanism randomness. Holding these objects fixed, it does not take adviser $i$'s own report or hardened advice as an input. This is \emph{conditional influence-freedom}. It permits the report to change the public action while preventing a second, direct path into the answer key. The deployed pipeline produces an \emph{operational} label $Y_i\in\calY_i$, which ideally coincides with $Y_i^\circ$ but in practice may deviate from it; \cref{def:leaky} quantifies the gap.

Finally, let
\[
    P=\pi_C(e(r),Z,X,U)
\]
be a provenance bundle and
\[
    A=\psi_C(Z,P)
\]
be the written resolution shown to humans.

Adviser $i$ receives transfer
\[
    t_i=G_i(e_i,Z,Y_i)+b_i(e_{-i},X,U),
\]
where $Y_i$ is the operational label and $b_i$ is independent of $r_i$. Adviser $i$ may also care directly about the public decision through $w_i(Z,\theta)$. Utility is
\[
    u_i=t_i+w_i(Z,\theta).
\]
The baseline information-provider case has $w_i\equiv 0$. The more general term covers a vendor who wants to win, an employee who prefers one project, or a forecaster with an ideological stake in the principal's action.

\begin{definition}[Settlement factorization]\label{def:factorization}
A mechanism is \emph{settlement-factorized} if:
\begin{enumerate}[leftmargin=1.5em]
    \item reports affect the public decision record and the reference labels only through hardened advice $e$;
    \item the public decision record $Z$ may depend on all hardened advice;
    \item for every adviser $i$, the reference-label map excludes $e_i$ as a direct input after conditioning on $(Z,e_{-i},X_i,U_i)$;
    \item transfers take the form $G_i(e_i,Z,Y_i)+b_i(e_{-i},X,U)$, where $Y_i$ is the operational label; and
    \item the written resolution is generated from $(Z,P)$ rather than serving as an independent payoff-relevant settlement channel.
\end{enumerate}
The mechanism is \emph{exactly} settlement-factorized if in addition $Y_i=Y_i^\circ$ almost surely for every adviser $i$, and \emph{approximately} settlement-factorized with leakage $(\varepsilon_1,\ldots,\varepsilon_n)$ if each operational label $Y_i$ is $\varepsilon_i$-leaky relative to $Y_i^\circ$ in the sense of \cref{def:leaky}.
\end{definition}

This is a structural definition, not merely a prompt-format convention. Calling two model invocations ``decision'' and ``evaluation'' does not factorize the mechanism if the evaluator receives the paid report, a summary written by that report, or a mutable rubric that the report can influence.

\begin{definition}[Provenance-linked writing]
A writer $A=\psi_C(Z,P)$ is provenance-linked if the memo is generated from the public decision and explicit supporting provenance. Operationally, there should exist a readout map $\chi$ such that $\Prb[\chi(A)=Z]$ is high and every payoff-relevant claim in $A$ is traceable to $P$.
\end{definition}

\subsection{Hardening invariance}

\begin{proposition}[Hardening invariance]\label{prop:hardening}
Suppose the mechanism is exactly settlement-factorized, so that every downstream map---including the operational label and the payment---depends on reports only through hardened advice. Fix adviser $i$, signal $s_i$, and reports $r_i,\widetilde r_i\in T_i^C(s_i)$. For every $r_{-i}$, replacing $r_i$ by $\widetilde r_i$ leaves the public decision, labels, transfers, and written resolution unchanged. Thus all faithful reporting policies are strategically equivalent.
\end{proposition}

The proposition is simple but institutionally important. It turns paraphrase robustness from a vague desideratum into a type-correct requirement: semantically equivalent reports must map to the same official advice record.

\begin{remark}[Approximate hardening invariance]\label{rem:approx-hardening}
Under approximate factorization the operational label may respond to the raw report beyond its hardened content, so exact invariance can fail. It degrades gracefully: two faithful reports induce the same $Z$ and the same reference-label law, so each operational label law lies within $\varepsilon_i$ of the same reference law, and by \cref{lem:tv-expectation} the two expected payments differ by at most $2L_i\varepsilon_i$.
\end{remark}

\subsection{Approximate factorization}

Exact conditional influence-freedom may be expensive. An operational evaluator may receive an imperfectly redacted transcript, use a model fine-tuned on the adviser's earlier outputs, or infer the report from correlated traces. Let $Y_i$ denote the operational label. Throughout, total variation is $\TV(\mu,\nu)=\sup_{A}|\mu(A)-\nu(A)|$, i.e., half the $L^1$ distance between densities; under this convention $\TV(\mathrm{Bernoulli}(a),\mathrm{Bernoulli}(b))=|a-b|$.

\begin{definition}[$\varepsilon_i$-leaky label]\label{def:leaky}
Fix a conditionally influence-free reference label $Y_i^\circ$. The operational label $Y_i$ is $\varepsilon_i$-leaky if, for every signal profile $s$, every report profile $r$ with hardened advice $e=e(r)$, and almost every $z$ in the support of $\Law(Z\mid s,r)$,
\[
\TV\!\left(
\Law(Y_i\mid s,r,Z=z),
\Law(Y_i^\circ\mid s,e_{-i},Z=z)
\right)
\le \varepsilon_i,
\]
where the conditional laws are regular conditional probabilities and the requirement is an essential supremum over positive-probability conditioning events.
\end{definition}

\begin{remark}[Endogeneity of the conditioning]\label{rem:leaky-endogeneity}
Because $Z$ is endogenous in the report, the event $\{Z=z\}$ refers to different underlying events under different report profiles. The definition therefore compares the operational and reference label kernels scenario by scenario, on the support of $Z$ induced by that scenario; the proof of \cref{thm:leakage} conditions on the realized $z$ within each scenario separately, which is legitimate because $Z$ is the same measurable object in both systems. Conditioning on the full report profile $r$, rather than only on hardened advice, is what allows the operational label to depend on presentation and other raw-report features beyond $e_i$---exactly the contamination $\varepsilon_i$ is meant to measure.
\end{remark}

The quantity $\varepsilon_i$ measures direct own-label contamination after conditioning on the public decision and the realized information state. It does not count the intended path $e_i\to Z$. If the action selected by $Z$ changes which outcome is observed, that decision-coupling problem belongs in the ideal externalized margin developed next.

\section{Truthful advice under exact factorization}
\label{sec:exact}

\subsection{The externalized advice margin}

Fix a realized signal profile $s$. Let $r_i^T\in T_i^C(s_i)$ be any faithful report and let $r_i'$ be a deviation. With the other advisers faithful, write $Z^T$ and $Z^D$ for the public decision records induced by $r_i^T$ and $r_i'$. Let $Y_i^{\circ,T}$ and $Y_i^{\circ,D}$ be the corresponding externalized labels.

Define the ideal score margin
\begin{align}
 m_i^\circ(s,r_i')
 =\E\big[&G_i(\phi_i^C(s_i),Z^T,Y_i^{\circ,T})
          -G_i(h_i^C(r_i'),Z^D,Y_i^{\circ,D})\mid s\big].
\label{eq:score-margin}
\end{align}

Define the deviation's direct decision gain
\begin{align}
 d_i(s,r_i')
 =\E\big[w_i(Z^D,\theta)-w_i(Z^T,\theta)\mid s\big].
\label{eq:decision-gain}
\end{align}

The total externalized truthful-advice margin is
\[
    \Gamma_i^\circ(s,r_i')=m_i^\circ(s,r_i')-d_i(s,r_i').
\]
This object deliberately absorbs the hard parts of the upstream mechanism. If advice changes the selected action and therefore changes the distribution of observable outcomes, that effect appears in $m_i^\circ$. If the adviser cares directly about winning or policy, that appears in $d_i$.

\begin{proposition}[Exact factorization preserves truthful advice]\label{prop:exact}
Suppose the mechanism is exactly settlement-factorized and, for every adviser $i$, signal profile $s$, and deviation $r_i'$, one has
\[
    \Gamma_i^\circ(s,r_i')\ge 0.
\]
Then every faithful reporting profile is a signal--ex-post equilibrium (\cref{rem:solution}). If the inequality is strict for every deviation that changes hardened advice, the equilibrium is strict up to hardening equivalence.
\end{proposition}

\begin{remark}[Solution concept]\label{rem:solution}
The equilibrium notion here and in \cref{thm:leakage} is \emph{ex post in signals and interim in mechanism randomness}: faithful reporting is a best response conditional on every realized signal profile $s$, in expectation over the labels and the internal randomness $(X,U,X_i,U_i)$. This is stronger than Bayesian Nash---no reliance on the prior over opponents' signals---and weaker than full ex-post optimality against realized randomness. We call it a \emph{signal--ex-post equilibrium}.
\end{remark}

The result is modular by design. Settlement factorization does not create $\Gamma_i^\circ$. It prevents the implementation from replacing the external evaluator with one partly written by the paid report.

\subsection{Proper scoring as one source of margin}

Consider first an action-neutral or action-complete verification object $V_i\in\mathcal V_i$ whose law is not changed by adviser $i$'s report. The latter includes a vector of potential outcomes or a proxy constructed before the report, even if the principal later selects one action. Let $q_i:\calE_i\to\Delta(\mathcal V_i)$ be the forecast encoded by hardened advice and let $S_i$ be a strictly proper scoring rule. The following assumption isolates the invariance that conditional influence-freedom alone does not supply, because $Y_i^\circ$ may depend on $Z$ and $Z$ depends on $e_i$.

\begin{assumption}[Action-neutrality]\label{ass:neutrality}
The unconditional law of the reference label is invariant to adviser $i$'s report: $\Law(Y_i^\circ\mid s)$ does not depend on $r_i$. Equivalently, any dependence of $Y_i^\circ$ on the public record routes only through components of $Z$ that are themselves invariant to $e_i$---as with an action-complete potential-outcome vector, a pre-report proxy, or a decision taken before the report is read.
\end{assumption}

\begin{corollary}[Proper scoring under an externalized label]\label{cor:proper}
Suppose
\[
    G_i(e_i,Z,Y_i^\circ)=S_i(q_i(e_i),Y_i^\circ)
\]
and truthful hardened advice posterior-implements the externalized label:
\[
    q_i(\phi_i^C(s_i))=\Law(Y_i^\circ\mid s_i).
\]
Under \cref{ass:neutrality} and $w_i\equiv 0$, every faithful reporting profile is a Bayesian Nash equilibrium. It is strict up to hardening equivalence whenever every evidence-changing deviation changes the induced forecast on a positive-probability set.
\end{corollary}

The report may contain much more than the forecast used by the score. Textual scoring mechanisms provide ways to map rich reports into proper numerical scores \parencite{wu2024elicitationgpt,lu2025aligned}; peer-prediction mechanisms can replace direct labels with correlated reports \parencite{miller2005peer,lu2024informative}. The corollary states the condition these modules must satisfy inside the factorized architecture.

\subsection{Advice value as score divergence and information}

For a proper score $S$, define the associated score divergence
\[
    D_S(p,q)=\E_{Y\sim p}[S(p,Y)-S(q,Y)]\ge 0.
\]
For regular differentiable proper scores this is the Bregman divergence associated with the score's convex entropy \parencite{gneiting2007scoring}.

\begin{proposition}[Information margin of faithful advice]\label{prop:info}
Let $S_i$ be a proper score for external label $Y_i^\circ$. Let the full signal $s_i$ induce posterior $p_{s_i}=\Law(Y_i^\circ\mid s_i)$, and let a deviation retain only a garbling $R_i$ of $s_i$---either a deterministic function $R_i=g_i(s_i)$, or a stochastic garbling satisfying the Markov condition $Y_i^\circ \ind R_i \mid s_i$---with posterior $p_{R_i}=\Law(Y_i^\circ\mid R_i)$. Then the expected score loss from the garbling is
\[
    \E\big[S_i(p_{s_i},Y_i^\circ)-S_i(p_{R_i},Y_i^\circ)\big]
    =\E\big[D_{S_i}(p_{s_i},p_{R_i})\big].
\]
Under the logarithmic score, this margin equals
\[
    I(Y_i^\circ;s_i\mid R_i).
\]
\end{proposition}

This gives the externalized margin an interpretable scale. A report is valuable to the extent that it retains information about the externalized evaluation object beyond what a coarser or strategic report reveals. Under log scoring, adding a second authorized label component $V$ to an existing label $Y$ contributes
\[
    I((Y,V);s_i\mid R_i)
    =I(Y;s_i\mid R_i)+I(V;s_i\mid Y,R_i).
\]
Thus a hybrid label improves the truthful margin exactly when it contains additional signal-relevant information, not merely additional text.

\subsection{Scoring recommendations rather than only forecasts}

Advice often culminates in a recommendation. Let $\alpha_i(e_i)\in\mathcal A$ be the action recommended by hardened advice and let $v(a,y)$ be the principal's value from action $a$ under externalized state label $y$.

\begin{corollary}[Decision-value score]\label{cor:decision}
Suppose
\[
    G_i(e_i,Z,Y_i^\circ)=v(\alpha_i(e_i),Y_i^\circ),
\]
$w_i\equiv 0$, and faithful advice recommends an action maximizing expected principal value:
\[
    \alpha_i(\phi_i^C(s_i))\in
    \arg\max_{a\in\alpha_i(\calE_i)}
    \E[v(a,Y_i^\circ)\mid s_i].
\]
Under \cref{ass:neutrality}, faithful advice is a Bayesian best response.
\end{corollary}

This score evaluates the recommendation against an external outcome whether or not the principal follows it.

\subsection{When action-neutrality fails: the decision-market boundary}
\label{sec:impossibility}

\Cref{ass:neutrality} fails mechanically whenever the principal deterministically follows the advice and only the chosen action's outcome is observed: the label is then a function of the report, and ordinary proper-scoring logic breaks because the advice changes the observation process. This is not a shortcoming of the factorized architecture but a known boundary. No symmetric scoring rule is quasi-strictly proper for the deterministic max decision rule \parencite{othman2010decision}, and strict properness for decision making requires stochastic decision rules with full support over actions \parencite{chen2011decision,chen2014decision}. The routes around this boundary identified by that literature---full-support randomization \parencite{chen2014decision}, compensation rules \parencite{boutilier2011experts}, proxy or pre-decision outcomes \parencite{wang2021decision,wang2023proxy}, performative fixed points \parencite{oesterheld2023performative}, and zero-sum joint scores \parencite{hudson2024joint}---are, in our terms, interchangeable modules for constructing an action-neutral or action-complete $Y_i^\circ$. The results above are agnostic to which module is used: the margin it creates enters as $m_i^\circ$, and its residual implementation leakage as $\varepsilon_i$. Settlement factorization is downstream of that choice.

\section{Approximate factorization and the price of leakage}
\label{sec:leakage}

This section carries the paper's core results, and their order tells one story. First the price: an $\varepsilon_i$-leaky label costs truthful margins at most $2L_i\varepsilon_i$, and a matching construction shows the constant is exact---half own-report manipulation, half pandering to a biased evaluator. Then two inversions. Worst-case incentive erosion over a canonical probe class \emph{identifies} the leakage, so $\varepsilon_i$ is a behavioral quantity, not merely an architectural one; and a ghost-reference construction shows that every mechanism, in every presentation, already carries an influence-free reference---there are no unfactorized mechanisms, only unmeasured ones; it is the revelation-principle move for this setting. A redundant-information family then makes the title quantitative: settlement on the shared record dilutes margins exponentially while a factorized label keeps a constant one. The section closes with the levers---stakes, audits, and the principal's purchasing problem---that turn the price schedule into design.

\subsection{Label sensitivity and the leakage theorem}

\begin{lemma}[TV--expectation bound]\label{lem:tv-expectation}
Let $f:\calY\to[a,b]$ be measurable and let $\mu,\nu$ be probability measures on $\calY$. Then
$\left|\E_{Y\sim\mu}[f(Y)]-\E_{Y\sim\nu}[f(Y)]\right|\le (b-a)\,\TV(\mu,\nu)$.
\end{lemma}

\begin{definition}[Label sensitivity]\label{def:stability}
The payment kernel $G_i$ is $L_i$-stable in label space if, for every $(e_i,z)$ and probability measures $\mu,\nu$ on $\calY_i$,
\[
\left|
\E_{Y\sim\mu}[G_i(e_i,z,Y)]-
\E_{Y\sim\nu}[G_i(e_i,z,Y)]
\right|
\le L_i\TV(\mu,\nu).
\]
If $G_i(e_i,z,y)\in[a_i,b_i]$, one may take $L_i=b_i-a_i$ by \cref{lem:tv-expectation}. This includes bounded Brier scores and clipped log scores; unbounded kernels such as the unclipped log score are not $L$-stable for any finite $L$, so boundedness is a substantive requirement.
\end{definition}

\begin{theorem}[Approximate settlement factorization preserves advice margins]\label{thm:leakage}
Fix adviser $i$. Suppose the mechanism has a conditionally influence-free reference label $Y_i^\circ$, an $\varepsilon_i$-leaky operational label $Y_i$, and an $L_i$-stable payment kernel. For every signal profile $s$ and deviation $r_i'$, let $\Gamma_i(s,r_i')$ be the truthful utility margin under the operational label and $\Gamma_i^\circ(s,r_i')$ the corresponding margin under the reference label. Then
\[
    \Gamma_i(s,r_i')
    \ge
    \Gamma_i^\circ(s,r_i')-2L_i\varepsilon_i.
\]
Consequently, if
\[
    \gamma_i:=\inf_{s,r_i':\,h_i^C(r_i')\ne\phi_i^C(s_i)}
    \Gamma_i^\circ(s,r_i')
    >2L_i\varepsilon_i,
\]
faithful reporting remains a strict signal--ex-post equilibrium up to hardening equivalence.
\end{theorem}

The factor of two appears because both sides of the incentive comparison can move: the operational label under faithful advice may differ from its reference, and the operational label under a deviation may differ from its own reference. \Cref{prop:tightness} below shows this is not an artifact of the proof: both movements can be realized simultaneously and adversarially, so the constant is best possible.

A useful decomposition separates score quality from direct decision interests. Let
\[
\underline m_i
=\inf_{s,r_i'}m_i^\circ(s,r_i')
\quad\text{and}\quad
D_i=\sup_{s,r_i'}d_i(s,r_i'),
\]
noting that $D_i\ge 0$ since $r_i'$ ranges over all reports, including faithful ones.

\begin{corollary}[Decision interests plus leakage]\label{cor:interests}
If $\underline m_i>D_i+2L_i\varepsilon_i$, faithful advice is a strict signal--ex-post equilibrium up to hardening equivalence.
\end{corollary}

This condition connects the paper to classical scoring rules for decision makers. Compensation or joint-scoring mechanisms can increase $\underline m_i$ or reduce the effective $D_i$; factorization reduces $\varepsilon_i$. These are distinct design levers, and the remainder of this section makes each of them quantitative.

\subsection{Tightness of the constant}

\begin{proposition}[Binary own-label manipulation]\label{prop:brier}
For every $\varepsilon\in(0,1/2)$ there is a one-adviser binary-label environment with reference label $Y^\circ\sim\mathrm{Bernoulli}(1/2)$ and Brier payment
\[
    G(q,y)=-(q-y)^2
\]
such that truthful reporting induces operational label $Y^\circ$, while a deviation reports $q'=1/2+\varepsilon$ and shifts the operational label to $Y'\sim\mathrm{Bernoulli}(1/2+\varepsilon)$. The label-law shift has total variation $\varepsilon$, and the deviation increases expected payment by exactly $\varepsilon^2$.
\end{proposition}

The Brier example establishes the basic point: once a report can move its own evaluation target, strict propriety against a fixed target no longer protects the institution. Its gain, however, is quadratic in $\varepsilon$, because the expected Brier score is locally flat in the label law at the truthful report. The linear rate of \cref{thm:leakage} is attained by payment kernels that respond linearly to the label law, and the constant $2$ is realized by combining own-report manipulation with a fixed evaluator bias.

\begin{proposition}[The leakage bound is exactly tight]\label{prop:tightness}
Fix $\varepsilon\in(0,1/2)$. Consider a one-adviser binary-label environment with reference label $Y^\circ\sim\mathrm{Bernoulli}(1/2)$, trivial public record, $w\equiv 0$, and the calling kernel
\[
    G(e,y)=y\,\mathbf 1\{q(e)>1/2\}+(1-y)\,\mathbf 1\{q(e)\le 1/2\}\in[0,1],
\]
so that $L=1$ by \cref{lem:tv-expectation} and every report has reference margin $\Gamma^\circ=0$.
\begin{enumerate}[label=(\alph*),leftmargin=1.9em]
    \item \textbf{Own-report manipulation.} If the operational label equals $Y^\circ$ unless the hardened report declares $q>1/2$, in which case $Y\sim\mathrm{Bernoulli}(1/2+\varepsilon)$, then $Y$ is $\varepsilon$-leaky and the deviation $q'=1/2+\varepsilon$ satisfies
    \[
        \Gamma(s,r')=\Gamma^\circ(s,r')-L\varepsilon.
    \]
    \item \textbf{Pandering to a biased evaluator.} If the operational label is $Y\sim\mathrm{Bernoulli}(1/2+\varepsilon)$ regardless of the report, then $Y$ is $\varepsilon$-leaky and the same deviation satisfies
    \[
        \Gamma(s,r')=\Gamma^\circ(s,r')-2L\varepsilon.
    \]
\end{enumerate}
Hence the constant $2$ in \cref{thm:leakage} is best possible, with the truthful side and the deviation side of the comparison each accounting for $L\varepsilon$.
\end{proposition}

\begin{remark}[What $\varepsilon$ measures]\label{rem:centered}
\Cref{def:leaky} measures total deviation of the operational label from the precommitted reference pipeline. \Cref{prop:tightness}(b) shows this bundles two strategically relevant components: report-caused contamination and a fixed evaluator bias that an adviser can pander to. A centered variant---bounding the total variation between the operational label laws induced by any two of adviser $i$'s own reports---isolates the manipulation component and, by a one-sided version of the argument for \cref{thm:leakage}, degrades margins by at most $L_i\varepsilon_i$ relative to the operational-truthful benchmark; but it leaves calibration to the reference uncontrolled. Institutions precommit to reference pipelines, not to their own biases, so we keep the uncentered definition.
\end{remark}

\subsection{A behavioral converse: erosion identifies leakage}
\label{sec:converse}

\Cref{thm:leakage} says leakage bounds erosion; \cref{prop:tightness} says the bound binds somewhere. Together they say more: worst-case erosion \emph{recovers} $\varepsilon_i$ exactly, so leakage---defined so far through pipeline architecture---has an equivalent definition purely in terms of incentives.

\begin{proposition}[Erosion identifies leakage]\label{prop:converse}
Fix adviser $i$'s operational label channel and its reference; for each hardened record $e$ let $\varepsilon_i(e)$ be the conditional total-variation gap of \cref{def:leaky} at input $e$, so that $\varepsilon_i=\sup_e\varepsilon_i(e)$, and let
\[
    E_i^\star\;:=\;\sup_{e\ne e'}\bigl[\varepsilon_i(e)+\varepsilon_i(e')\bigr].
\]
Consider the canonical class of probes: one-adviser environments with trivial public record and $w\equiv0$ in which any ordered pair of \emph{distinct} hardened records $(e,e')$ is realizable as (faithful, deviation), equipped with any $L$-stable payment kernel. Then
\[
    \sup\;\bigl[\Gamma_i^\circ(s,r_i')-\Gamma_i(s,r_i')\bigr]\;=\;L\,E_i^\star,
    \qquad\text{with}\qquad
    L\,\varepsilon_i\ \le\ L\,E_i^\star\ \le\ 2L\,\varepsilon_i ,
\]
the supremum taken over the class. If the supremal leakage is approached along at least two distinct hardened records---in particular whenever a report-independent evaluator-bias component is present, as in \cref{prop:tightness}(b)---then $E_i^\star=2\varepsilon_i$ and worst-case erosion identifies leakage exactly. Consequently a mechanism whose margins erode by at most $2L\varepsilon'$ across the class is $2\varepsilon'$-leaky in general and $\varepsilon'$-leaky under the two-record condition; the converse direction is \cref{thm:leakage}. Within this class, behavioral robustness and architectural leakage are the same property up to a factor of at most two, and exactly the same under the richness condition.
\end{proposition}

\begin{proof}
\emph{Upper bound.} For any probe with faithful record $e$ and deviation record $e'$, decompose the erosion as
\[
    \Gamma^\circ-\Gamma
    =\bigl[\E G(e,Y^\circ)-\E G(e,Y_e)\bigr]
    +\bigl[\E G(e',Y_{e'})-\E G(e',Y^\circ)\bigr],
\]
and bound each bracket by \cref{lem:tv-expectation}: the first compares two label laws at total variation at most $\varepsilon_i(e)$ under a fixed kernel argument, the second at most $\varepsilon_i(e')$. Hence every probe erodes by at most $L\bigl(\varepsilon_i(e)+\varepsilon_i(e')\bigr)\le L\,E_i^\star$.

\emph{Lower bound.} Fix distinct records $e,e'$ and $\delta>0$. By the definition of total variation as a supremum over events, choose $A$ with $\Prb_{Y|e}(A)-\Prb_{Y^\circ}(A)\ge\varepsilon_i(e)-\delta$ and $A'$ with $\Prb_{Y|e'}(A')-\Prb_{Y^\circ}(A')\ge\varepsilon_i(e')-\delta$. Define the kernel
\[
    G(e'',y)=
    \begin{cases}
        L\,\mathbf 1\{y\notin A\} & e''=e,\\
        L\,\mathbf 1\{y\in A'\} & e''=e',
    \end{cases}
\]
which has range $[0,L]$ and is $L$-stable by \cref{lem:tv-expectation}. The same decomposition evaluates to
\[
    L\bigl[\Prb_{Y_e}(A)-\Prb_{Y^\circ}(A)\bigr]
    + L\bigl[\Prb_{Y_{e'}}(A')-\Prb_{Y^\circ}(A')\bigr]
    \;\ge\; L\bigl(\varepsilon_i(e)+\varepsilon_i(e')-2\delta\bigr).
\]
Taking suprema over distinct pairs and $\delta\to0$ gives $L\,E_i^\star$, so the supremum equals it. For the sandwich: one slot can carry a near-supremal record while the other contributes nonnegatively (the record space contains at least two elements), so $E_i^\star\ge\varepsilon_i$; and each slot contributes at most $\varepsilon_i$, so $E_i^\star\le2\varepsilon_i$, with equality when two distinct records approach the supremum.
\end{proof}

\begin{remark}[Architecture-free leakage, and what stays open]\label{rem:converse}
\Cref{prop:converse} gives $\varepsilon_i$ an operational meaning that never mentions the pipeline's internals: it is the largest per-stake incentive gap an adversarially chosen bounded score can extract, normalized by $2L$ (exactly under the two-record condition of \cref{prop:converse}, and within a factor of two in general). This is exactly the quantity the paired-counterfactual harness of \cref{sec:measurement} estimates on a finite probe menu, so the measurement program and the theory are estimating the same object by construction. Two honest limits. First, identification is relative to the declared reference pipeline: a mechanism that commits to no reference has no $\varepsilon$ to identify, and behavioral probing then measures only the centered manipulation component of \cref{rem:centered}. Second, the class quantifies over probe kernels; the stronger representation question---whether every mechanism that is robust under its own \emph{fixed} kernel admits some factorized representation with small leakage---is taken up next: \cref{sec:normalform} shows a reference always exists constructively, which reduces the question to a well-posed minimization.
\end{remark}

\subsection{A revelation principle for settlement: the ghost reference}
\label{sec:normalform}

Everything so far presumes an institution that has declared a reference pipeline. This subsection removes that presumption. An arbitrary mechanism need not come with a designated label; say a \emph{presentation} of adviser $i$'s settlement is any factorization of the transfer as $t_i=G_i(e_i,Z,Y_i)+b_i(e_{-i},X,U)$ in which $Y_i$ is produced by an operational channel
\[
    P_{e}\bigl(\,\cdot\mid z,e_{-i},x,u\bigr)\;=\;\Law\bigl(Y_i \mid e_i=e,\;z,e_{-i},x,u\bigr).
\]
Every mechanism has at least the trivial presentation $Y_i=t_i$ with $G_i$ the identity, and mechanisms built from scoring rules have their natural presentation with $Y_i$ the outcome label; the kernel slot $G_i(e_i,\cdot,\cdot)$ carries the \emph{authorized} own-report dependence (the forecast being scored), while the channel carries the contested one.

\begin{definition}[Intrinsic leakage of a presentation]\label{def:intrinsic}
The intrinsic leakage of the channel is
\[
    \varepsilon_i^\star\;:=\;\operatorname*{ess\,sup}_{(z,e_{-i},x,u)}\ \sup_{e,e'}\ \TV\bigl(P_{e},P_{e'}\bigr),
\]
the centered quantity of \cref{rem:centered}. Its definition mentions no reference.
\end{definition}

\begin{theorem}[Ghost reference; factorization is a normal form]\label{thm:ghost}
Let $\nu_i$ be any committed distribution over hardened records, measurable in authorized objects only, and define the \emph{ghost label} $Y_i^\circ$ by running the same operational channel on an independent resample $\widetilde E_i\sim\nu_i$ in place of $e_i$. Then:
\begin{enumerate}[label=(\alph*),leftmargin=1.9em]
    \item $Y_i^\circ$ is conditionally influence-free, so $(Y_i,Y_i^\circ)$ satisfies the label clause of \cref{def:factorization};
    \item $Y_i$ is $\varepsilon$-leaky with respect to this reference with $\varepsilon\le\varepsilon_i^\star$;
    \item every conditionally influence-free reference $R$ satisfies $\varepsilon_R\ge\varepsilon_i^\star/2$; hence, writing $\varepsilon_i^{\mathrm{opt}}$ for the infimum over references,
    \[
        \varepsilon_i^{\mathrm{opt}}\ \le\ \varepsilon_{\mathrm{ghost}}\ \le\ \varepsilon_i^\star\ \le\ 2\,\varepsilon_i^{\mathrm{opt}} ;
    \]
    \item \cref{thm:leakage} applies verbatim to the pair: $\Gamma_i\ge\Gamma_i^{\circ,\mathrm{ghost}}-2L_i\varepsilon_i^\star$.
\end{enumerate}
In particular, every advice mechanism, in every presentation, \emph{is} an approximately settlement-factorized mechanism with respect to a constructible reference whose leakage is within a factor of two of the best achievable.
\end{theorem}

\begin{proof}
(a) By construction the ghost channel's conditional law given $(z,e_{-i},x,u)$ does not involve $e_i$. (b) By joint convexity of total variation, $\TV\bigl(P_{e},\int P_{e'}\,\nu_i(de')\bigr)\le\int\TV(P_{e},P_{e'})\,\nu_i(de')\le\varepsilon_i^\star$. (c) For any two records, the triangle inequality through $R$ gives $\TV(P_e,P_{e'})\le\TV(P_e,R)+\TV(R,P_{e'})\le2\varepsilon_R$; take suprema. (d) The hypotheses of \cref{thm:leakage} are exactly (a) and (b).
\end{proof}

\begin{remark}[What the normal form buys, and the sharpened open problem]\label{rem:normalform}
Three consequences. First, no mechanism escapes the accounting: an institution cannot plead that no reference exists, because the ghost manufactures one from the mechanism's own channel and committed randomness. There are no unfactorized mechanisms, only unmeasured ones. This is the revelation-principle move for model-mediated settlement \parencite{myerson1979incentive}: just as arbitrary mechanisms are representable by direct ones, arbitrary settlement channels are representable---unchanged---as factorized channels carrying a measured leakage, so designing within the factorized class is without loss of generality. The choice of $\nu_i$ moves the reference margin $\Gamma^{\circ,\mathrm{ghost}}$ but not the leakage arithmetic, and leave-one-out labels---including \cref{prop:beyond}(b)---are the degenerate ghost that deletes the input slot. Second, $\varepsilon_i^\star$ is an \emph{invariant}: it is defined with no reference, and by (c) together with \cref{prop:converse} applied to the ghost reference, it coincides up to a factor of two with the best-reference leakage and with worst-case probe erosion normalized by $2L$. Robustness in this paper is therefore a single number attached to a mechanism, not a property of a favored architecture. Third, the open problem sharpens rather than dissolves: the mechanism-level invariant is $\min\varepsilon_i^\star$ over presentations---finding the presentation that minimizes intrinsic leakage is finding the mechanism's most-factorized representation, for which the ghost gives a constructive $2$-approximation presentation by presentation. The residual necessity gap is fixed-kernel blindness: a kernel insensitive to the influenced directions of the label law can be robust while $\varepsilon_i^\star$ is large, and closing that gap for natural kernel families (proper scores, monotone kernels) is open; \cref{sec:conjecture} states the full mechanism-level form as a conjecture.
\end{remark}

\subsection{Beyond self-resolution, quantified}
\label{sec:beyond}

The title promises an advance over settling advice on the very pipeline the advice steers. The promise can be made quantitative in the canonical redundant-information family of \textcite{dellapenna2024nlm}: the latent state is $\theta=A\oplus B$ with $A,B$ independent and uniform, and each variable is observed by $k$ advisers (odd $k$) through conditionally independent signals of accuracy $p\in(\tfrac12,1)$; the shared record $Z$ majority-decodes each variable and outputs the XOR.

\begin{proposition}[Shared settlement dilutes; factorized settlement does not]\label{prop:beyond}
In this family:
\begin{enumerate}[label=(\alph*),leftmargin=1.9em]
    \item \textbf{(Self-resolution.)} Consider any mechanism that settles every adviser on the shared pipeline: each $t_i$ is measurable in the pair $(Z,\theta)$ and has range at most $L$. Then every unilateral faithful margin obeys
    \[
        \Gamma_i\;\le\;L\,\binom{k-1}{\tfrac{k-1}{2}}\bigl[p(1-p)\bigr]^{\frac{k-1}{2}}\;=\;L\cdot\Theta\!\bigl((4p(1-p))^{k/2}/\sqrt{k}\bigr)\;\longrightarrow\;0 ,
    \]
    so at bounded stakes no shared-settlement payment sustains a fixed margin as redundancy grows, and deterring a fixed decision interest $D$ requires stakes growing exponentially in $k$.
    \item \textbf{(Factorization.)} The leave-one-out peer label $Y_i=$ majority of the other $k-1$ same-variable reports (fair coin on ties) is conditionally influence-free by construction ($\varepsilon_i=0$: it is measurable in $e_{-i}$ alone), and the unit-range agreement kernel $G_i=\mathbf 1\{e_i\text{ matches }Y_i\}$ delivers margin
    \[
        \gamma_i(k)\;=\;(2p-1)\bigl(2\mu_{k-1}-1\bigr)\;\longrightarrow\;2p-1\;>\;0,
    \]
    where $\mu_{k-1}\to1$ is the accuracy of the others' majority.
\end{enumerate}
Hence in this family factorization is not merely cheaper than stakes; among the two architectures it is the only one whose margins do not vanish as the crowd that informs the decision grows.
\end{proposition}

\begin{proof}
(a) Under a unilateral change of report, couple the two runs of the mechanism off the event that the deviator is pivotal for their variable's majority; the joint law of $(Z,\theta)$ then changes only on that event, whose probability is the stated binomial expression, so the laws are within that distance in total variation and \cref{lem:tv-expectation} bounds the payment difference---hence the margin---by $L$ times it. Stirling gives the asymptotics, and the stakes statement follows as in \cref{cor:stakes} with $D$ fixed. (b) Influence-freedom is definitional. Conditioning on the variable's value, the adviser's signal and the others' majority are independent, so $\Prb(e_i\text{ matches }Y_i)=p\mu_{k-1}+(1-p)(1-\mu_{k-1})$ under faith, and the flip margin is $2\Prb(\text{match})-1=(2p-1)(2\mu_{k-1}-1)$.
\end{proof}

Part (a) is an upper bound over \emph{all} kernels settled on the shared record and realized outcome; the exact achieved margin of the canonical outcome-anchored kernel in this family, $\eta(k)$, is computed in version~2 of \textcite{dellapenna2024nlm} and matches the same exponential rate. Part (b) is precisely a settlement-factorized mechanism, discovered from inside the self-resolution model's own best example: the escape from dilution is not a larger stake or a better prompt but the removal of the paid report from its own answer key.

\subsection{Stakes and audits as design levers}
\label{sec:levers}

The sufficient condition $\gamma_i>D_i+2L_i\varepsilon_i$ has three terms an institution can move: the margin created upstream, the adviser's outside interest in the decision, and the leakage tax. Two levers deserve explicit analysis because they behave very differently under scaling. The first is money.

\begin{corollary}[Stakes outbid decision interests, not evaluator capture]\label{cor:stakes}
Fix adviser $i$ and scale the payment kernel by $c>0$, so $G_i^c=cG_i$. Then $G_i^c$ is $cL_i$-stable, the ideal score margin scales to $c\,m_i^\circ(s,r_i')$, and the decision gain $d_i(s,r_i')$ is unchanged. Faithful advice is a strict signal--ex-post equilibrium under the operational label whenever
\[
    c\,(\underline m_i-2L_i\varepsilon_i)>D_i.
\]
Consequently:
\begin{enumerate}[label=(\alph*),leftmargin=1.9em]
    \item if $\underline m_i>2L_i\varepsilon_i$, then for every finite decision interest $D_i$ the stake multiplier $c^\ast=D_i/(\underline m_i-2L_i\varepsilon_i)$ restores faithful advice for all $c>c^\ast$;
    \item if $\underline m_i\le 2L_i\varepsilon_i$, the condition fails for every $c$, and the failure is real: in the environment of \cref{prop:tightness}(b), scaled by $c$, the pandering deviation's profit is exactly $2cL\varepsilon$, so raising stakes strictly increases the manipulation rent.
\end{enumerate}
\end{corollary}

The dimensionless ratio $\underline m_i/(2L_i\varepsilon_i)$ at normalized stakes is therefore an \emph{integrity ratio} for the settlement channel: incentive power is safely scalable if and only if it exceeds one. Information rents and manipulation rents both grow linearly in the stakes; only the architecture changes their ratio. The caveat is that $D_i$ must genuinely be outside the scaled transfer---a vendor whose contract value grows with the institution's budget scales $D_i$ along with $c$, and then no multiplier helps.

The second lever is randomized recourse to the reference pipeline itself, in the spirit of costly state verification and random auditing \parencite{townsend1979costly}.

\begin{proposition}[Audited settlement]\label{prop:audit}
Suppose that with probability $p_i$, committed in advance and drawn from mechanism randomness independent of $(s,r,Z)$, adviser $i$'s settlement uses the reference label $Y_i^\circ$, and otherwise uses the operational label $Y_i$. If $Y_i$ is $\varepsilon_i$-leaky, the mixture label is exactly $(1-p_i)\varepsilon_i$-leaky. Hence faithful advice is a strict signal--ex-post equilibrium whenever $p_i>p_i^\ast$, where
\[
    p_i^\ast:=\max\Bigl\{0,\ 1-\frac{\gamma_i-D_i}{2L_i\varepsilon_i}\Bigr\},
\]
with weak faithfulness at the threshold itself; and if reference settlement costs $k_i^\circ$ per episode against $k_i<k_i^\circ$ for the operational pipeline, the minimal expected settlement cost compatible with faithful advice is $k_i+p_i^\ast\,(k_i^\circ-k_i)$, the infimum over strictly faithful audit rates, attained weakly at the threshold.
\end{proposition}

\begin{remark}[Audits must be unpredictable]\label{rem:committed-audits}
Independence of the audit lottery from the scenario is essential, not decorative. If $p_i$ can depend on the report, an adviser deviates only where the audit probability is low, and the mixture bound fails scenario by scenario. The commitment can be implemented with a public randomness beacon or a precommitted seed; what matters is that the adviser's report cannot move its own audit probability. Deterrence variants---drawing both labels and penalizing divergence---require care, because operational and reference labels are stochastic and can disagree by chance; a penalty scheme needs repeated paired draws or verifiable pipeline traces to separate manipulation from noise.
\end{remark}

The audit rate $p_i^\ast$ prices integrity directly: the institution buys exactly the leakage reduction it lacks, and \cref{prop:tightness} shows the price cannot be negotiated below the $2L_i\varepsilon_i$ schedule in the worst case.

\subsection{Buying integrity: the principal's problem}
\label{sec:frontier}

The levers of this section and the certification of \cref{sec:certified} are best read together as a \emph{menu of settlement technologies}: each technology $T$ delivers a certified leakage $\varepsilon_T$ at an expected per-episode cost $k_T$. The institution's problem is then not to maximize integrity but to purchase it.

\begin{proposition}[Integrity is a threshold good]\label{prop:threshold}
Fix adviser primitives $(\gamma_i,D_i,L_i)$ and let $\varepsilon^{\mathrm{req}}_i:=(\gamma_i-D_i)/(2L_i)$, and suppose faithful advice is worth $V_i>0$ per episode to the principal.
\begin{enumerate}[label=(\alph*),leftmargin=1.9em]
    \item Any technology with $\varepsilon_T\ge\varepsilon^{\mathrm{req}}_i$ carries no faithfulness guarantee, and by \cref{prop:tightness} the failure is realizable, so leakage reduction that stops short of the requirement has zero marginal incentive value. The principal's program is therefore bang--bang: adopt the cheapest technology with $\varepsilon_T<\varepsilon^{\mathrm{req}}_i$ if its cost is below $V_i$, and otherwise do not pay for advice through this channel.
    \item On the audit frontier generated by \cref{prop:audit} from an operational pipeline $(\varepsilon_i,k_i)$ and its reference $(0,k^\circ_i)$---namely $\varepsilon(p)=(1-p)\varepsilon_i$ at cost $k_i+p\,(k^\circ_i-k_i)$---the program's infimum lies at the threshold $p=p_i^\ast$ (weak faithfulness there; any $p>p_i^\ast$ purchases strictness), and the shadow price of integrity is constant:
    \[
        \frac{k^\circ_i-k_i}{\varepsilon_i}\quad\text{per unit of total variation removed.}
    \]
    \item Between technologies the choice is a price comparison: architectural certification (\cref{sec:certified}) dominates auditing exactly when its cost is below $k_i+p_i^\ast(k^\circ_i-k_i)$.
\end{enumerate}
\end{proposition}

\begin{proof}
(a) Above the requirement, \cref{prop:tightness} exhibits an environment meeting the primitives in which a deviation is profitable, so no guarantee is available; below it, \cref{thm:leakage} guarantees faithfulness; the feasible set in $\varepsilon$ is the interval $[0,\varepsilon^{\mathrm{req}}_i)$ and value is flat on each side of its boundary, giving the bang--bang structure. (b) Cost is strictly increasing in $p$ and the constraint $\varepsilon(p)<\varepsilon^{\mathrm{req}}_i$ binds first at $p_i^\ast$; the frontier is affine with the stated slope. (c) is substitution of costs.
\end{proof}

The economics are deliberately stark. Integrity here is not a graded virtue but a constraint with a posted price; the tightness result says the requirement $\varepsilon^{\mathrm{req}}_i$ is non-negotiable in the worst case, and every quantity in the program---$\gamma_i$, $D_i$, $L_i$, $\varepsilon_i$, $k_i$, $k^\circ_i$---is an estimable primitive of \cref{sec:measurement}. An institution that cannot state its price per unit of total variation has not yet designed its settlement layer.

\section{Beyond verification: constructing externalized labels}
\label{sec:labels}

Settlement factorization does not require a directly observable ground truth. It requires an explicit answer to three questions: what object is being scored, how is that object produced, and which causal paths from adviser $i$ to its own label are authorized?

\Cref{tab:label-sources} organizes common choices.

\begin{table}[H]
\centering
\small
\begin{tabularx}{\textwidth}{>{\raggedright\arraybackslash}p{0.18\textwidth}>{\raggedright\arraybackslash}p{0.25\textwidth}>{\raggedright\arraybackslash}X>{\raggedright\arraybackslash}p{0.22\textwidth}}
\toprule
\textbf{Label source} & \textbf{Example} & \textbf{Advantage} & \textbf{Characteristic risk}\\
\midrule
Direct verification & Hidden test, realized delivery time, audited event & Clear semantics and ordinary proper scoring & Selective observability; advisers may influence which outcome is realized.\\
\addlinespace
Action-complete or randomized outcome & Potential-outcome vector, stochastic decision market & Separates prediction from action selection & Counterfactuals unavailable or costly; randomization can reduce principal utility.\\
\addlinespace
Proxy or holdout & Principal signal, independent audit, withheld benchmark & Deterministic decisions may remain possible \parencite{wang2023proxy} & Proxy misalignment and distribution shift.\\
\addlinespace
Peer or later report & Cross-prediction, reference adviser, self-resolving market & Works without direct verification \parencite{miller2005peer,srinivasan2023selfresolving} & Collusion, common-mode errors, reference manipulation.\\
\addlinespace
Evidence-based endogenous label & Evidence market, adjudicated dossier & Makes reasons payoff-relevant and can resolve otherwise unverifiable questions \parencite{hossain2026evidence} & Evidence withholding, judge manipulation, staking and verification errors.\\
\addlinespace
Hybrid settlement & Tests + audits + peers + external outcomes & Can add independent information and graceful fallback & Hidden double counting, correlated components, complex provenance.\\
\addlinespace
Joint or competitive label & Zero-sum scores across advisers & Can remove incentives to steer the principal under suitable assumptions \parencite{hudson2024joint} & Simultaneity, coalition behavior, and participation constraints.\\
\bottomrule
\end{tabularx}
\caption{Externalized labels need not be literal ground truth. The relevant requirement is a precommitted, auditable production rule whose remaining own-report influence can be measured.}
\label{tab:label-sources}
\end{table}

\subsection{Direct verification is the easy case}

When a target $Y^*$ is observed independently of the report, set $Y_i^\circ=Y^*$. The target and settlement label coincide, $\varepsilon_i$ can be made close to zero, and ordinary proper scoring applies. This is a useful baseline but not the domain that motivates natural language mechanisms.

\subsection{Self-resolution and peer labels}

When $Y^*$ is unavailable, a later or better-informed report may serve as a proxy \parencite{srinivasan2023selfresolving}. In the advice architecture, the reference-selection rule must itself be factorized. Adviser $i$ should not choose which later report is used, write the summary shown to the reference agent, or supply the rubric by which agreement is judged. Cross-fitting can help: use one partition of reports to form the public decision and another to construct labels, rotating roles across advisers. This is the mechanism-design analogue of cross-fitting in double/debiased machine learning, where sample splitting prevents a unit's own data from contaminating the nuisance estimate used to evaluate it \parencite{chernozhukov2018double}.

\subsection{Evidence and adjudication}

Evidence markets explicitly reward evidence and permit endogenous resolution when external truth is absent \parencite{hossain2026evidence}. Their evidence verification and staking machinery can instantiate the hardening and label-production layers of this paper. The remaining architectural question is whether trader $i$'s submitted evidence directly changes the label against which trader $i$ is paid, beyond the influence already accounted for by the evidence-market incentive theorem. A factorized implementation can use leave-one-out adjudication, independent challengers, or precommitted evidence-quality rules to estimate the residual $\varepsilon_i$.

\subsection{Hybrid labels}

Real institutions often have partial verification. A software task has unit tests but also qualitative requirements. A grant has publication outcomes but also long delays and selection effects. A moderation decision has policy text, peer judgments, and appeals. Hybrid settlement combines these sources.

The information-margin result in \cref{sec:exact} gives a disciplined reason to add a component: under log scoring, it should contribute conditional information about the adviser's signal after the existing label is known. The leakage theorem gives the opposing cost: every additional model-mediated component may increase $\varepsilon_i$. A larger label is not automatically a better label.

The leakage of a hybrid is controlled by its components only under an independence condition.

\begin{lemma}[Leakage composition for hybrid labels]\label{lem:composition}
Let $Y_i=h(C_1,\ldots,C_K)$ for a measurable map $h$, where each operational component $C_k$ is $\varepsilon_k$-leaky relative to a reference component $C_k^\circ$ in the componentwise sense of \cref{def:leaky}. Suppose that, conditionally on the leakage conditioning variables---$(s,r,Z=z)$ for the operational components and $(s,e_{-i},Z=z)$ for the reference components---the components $C_1,\ldots,C_K$ are mutually independent, as are $C_1^\circ,\ldots,C_K^\circ$. Then $Y_i$ is $\bigl(\sum_{k=1}^K\varepsilon_k\bigr)$-leaky relative to $Y_i^\circ=h(C_1^\circ,\ldots,C_K^\circ)$.
\end{lemma}

\begin{remark}[Independence is essential]\label{rem:xor}
Let $C_1^\circ,C_2^\circ$ be independent uniform bits, and let the operational components be $C_1$ uniform and $C_2=C_1\oplus\mathbf 1\{\text{deviation}\}$. Each operational component is conditionally identical in law to its reference ($\varepsilon_1=\varepsilon_2=0$ componentwise), yet the pair $(C_1,C_2)$ determines whether the adviser deviated: the hybrid $h(C_1,C_2)=C_1\oplus C_2$ equals $\mathbf 1\{\text{deviation}\}$ exactly. Per-component audits therefore certify nothing about a hybrid label unless independence is enforced by construction---separate randomness, separate models, disjoint contexts---or the joint law is audited directly.
\end{remark}

\subsection{Certified factorization by construction}
\label{sec:certified}

Estimating $\varepsilon_i$ statistically, as in \cref{sec:measurement}, yields lower bounds and falsification: a classifier that distinguishes the operational label law from the reference law certifies that leakage is at least some amount, but over rich label spaces no feasible sample certifies that leakage is at most $\varepsilon_i$. Upper bounds must come from the architecture. Differential privacy supplies exactly the right currency: a pipeline that is differentially private in the paid adviser's report slot is leakage-bounded by construction, with the added benefit that its guarantees compose adaptively, where \cref{lem:composition} could not.

\begin{proposition}[Differential privacy certifies leakage]\label{prop:dp}
Fix adviser $i$ and define the reference pipeline as the operational pipeline run with adviser $i$'s report replaced by a fixed admissible placeholder $\bot_i$ (a leave-one-out reference---the degenerate ghost of \cref{thm:ghost}). Suppose that, for every conditioning scenario of \cref{def:leaky}, the operational label pipeline is $\delta_i$-differentially private in adviser $i$'s report slot: for all reports $r_i,\widetilde r_i$, all $r_{-i}$, and all measurable $A$,
\[
    \Prb\bigl(Y_i\in A\mid s,(r_i,r_{-i}),Z=z\bigr)
    \le e^{\delta_i}\,
    \Prb\bigl(Y_i\in A\mid s,(\widetilde r_i,r_{-i}),Z=z\bigr).
\]
Then $Y_i$ is $(1-e^{-\delta_i})$-leaky, hence $\delta_i$-leaky. Moreover:
\begin{enumerate}[label=(\alph*),leftmargin=1.9em]
    \item \textbf{Post-processing.} Any measurable function of $Y_i$ and report-independent randomness satisfies the same bound.
    \item \textbf{Adaptive composition.} If the label is $h(C_1,\ldots,C_K)$, where each component is produced with independent internal randomness by a stage that is $\delta_{i,k}$-DP in adviser $i$'s report slot for every realization of the earlier components, then the hybrid label is $\bigl(1-e^{-\sum_k\delta_{i,k}}\bigr)$-leaky---even when later components are chosen adaptively as functions of earlier ones.
\end{enumerate}
\end{proposition}

\begin{remark}[Why the DP route survives adaptivity]\label{rem:dp-adaptive}
\Cref{lem:composition} requires conditional independence across components, and \cref{rem:xor} shows the requirement cannot be dropped: marginally clean components can jointly encode the report. Differential privacy is immune to that construction because it is a pointwise likelihood-ratio bound on each stage's \emph{conditional} law given everything produced so far; the XOR stage $C_2=C_1\oplus\mathbf 1\{\text{deviation}\}$ has an infinite likelihood ratio and is not $\delta$-DP for any finite $\delta$. The standard privacy toolbox \parencite{dwork2014roth} thus transfers wholesale: composition budgets $\varepsilon_i$ across pipeline stages, post-processing is data-processing invariance, and randomized audit settlement (\cref{prop:audit}) is privacy amplification by subsampling in mechanism-design clothing.
\end{remark}

Architecturally, the proposition says an institution can \emph{purchase a certificate}: inject calibrated randomness (randomized response over label bins, noisy score aggregation) or enforce separation (redaction with verified provenance, cross-model evaluation with committed weights, trusted execution), account the resulting $\delta_i$ per stage, and publish the budget. The measurement program of \cref{sec:measurement} then audits the certificate from below while the DP accounting bounds it from above. The cost is informational---noise in the label reduces the very margins of \cref{prop:info}---and pricing that tradeoff exactly (for example, the margin cost of randomized-response noise as a function of the certified budget) is left to future work.

\section{What is specific to language?}
\label{sec:language}

Most of the mathematical results concern information, decisions, and settlement. They would still apply to structured evidence vectors. Language matters because it defines the interface through which those objects are created and attacked.

\paragraph{Constitutions are textual.}
Admissibility, objectives, exceptions, and tie-breaking rules are written in words. Nearby rewrites can change which distinctions a model notices. Constitution sensitivity is therefore a mechanism property, not merely a prompt-engineering detail.

\paragraph{Hardening is an institutional act.}
The hardener decides which claims become official, what support is attached, and which instructions are treated as data rather than commands. Indirect prompt injection is a failure of this boundary \parencite{greshake2023indirect,liu2023prompt}. So are strategic presentation effects that alter machine decisions without adding corresponding human-relevant information. The latter are formalized as mecha-nudges by \textcite{frey2026mecha}.

\paragraph{Persuasion and information are not identical.}
An unrestricted report can increase a model's propensity to choose an action without increasing the externally measurable quality of that action. Canonicalization, claim extraction, adversarial redaction, and provenance can reduce this gap, but they can also discard useful nuance. The hardening map is therefore part of the mechanism and should be evaluated for both incentive compatibility and information loss.

\paragraph{Written resolutions can fork institutional reality.}
If the memo is generated independently after payments are determined, it may explain a decision the institution did not actually settle or introduce unsupported findings. Provenance-linked writing makes the public text a readout of $Z$, not a second adjudicator.

\paragraph{Module names do not establish independence.}
Separate prompts to the same model can share biases, hidden state, retrieval context, or training contamination. Operational factorization may require access controls, secret holdouts, cross-model evaluation, role separation, or cryptographic commitments. The prompt-injection defense literature supplies concrete architectures for such separation \parencite{debenedetti2025camel,beurer2025design}, and documented self-preference of LLM judges toward their own generations \parencite{panickssery2024llm} shows why cross-model evaluation in particular matters. The theorem prices the resulting distributional leakage; it does not certify a software diagram.

\section{Limits and complementary layers}
\label{sec:limits}

Settlement factorization addresses one failure mode: evaluator capture by the report being evaluated. It does not imply that the external evaluator is correct, fair, coherent, or coalition-proof.

\paragraph{Normative correctness.}
A factorized mechanism can faithfully implement a bad constitution. Empirical work finds substantial gaps between unconstrained LLM allocations and distributive-fairness criteria, with menu-based choice sometimes performing better than free-form generation \parencite{hosseini2025fairness}. Welfare, fairness, and legal constraints should therefore be encoded explicitly in the public decision rule rather than inferred from evaluator independence.

\paragraph{Identity and clone attacks.}
An agent may create multiple near-identical providers. Clone-robust ranking mechanisms address this entry-layer problem \parencite{hays2026candidacy}. Leave-one-out labels indexed by nominal identity are ineffective when one economic actor controls several identities.

\paragraph{Coalition payments.}
Multi-agent production creates a separate credit-assignment problem. Least-core methods can divide value among contributors in a coalition-stable way \parencite{harris2026core}. Such a rule can sit downstream of factorized coalition-value estimates, but factorization alone does not choose the division.

\paragraph{Collusion and orchestration.}
Repeated advisers may coordinate reports, prices, or mutual evaluations. Meta-game approaches distinguish adaptation among deployed agents from coordination induced by common designers or shared parameterization \parencite{carissimo2025orchestration,luo2026collusion}. A settlement pipeline can be individually influence-free and still coalition-manipulable.

\paragraph{Behavioral implementation.}
Mechanism guarantees describe best responses, not whether a particular language model finds them. Work on reasoning agents in repeated games studies when off-the-shelf agents approach equilibrium behavior \parencite{kang2026reasoning}; controlled prediction-market experiments show that LLM information aggregation degrades with structural complexity even when markets remain relatively robust to several interventions \parencite{galanis2026aggregation}. These are complementary empirical questions.

\section{Measuring the design variables}
\label{sec:measurement}

The theory suggests measuring structural quantities rather than only final task accuracy.

\paragraph{Externalized advice margin $\gamma_i$.}
Estimate the payoff difference between faithful advice and strategically chosen alternatives when labels are generated by the reference pipeline. For log-scored forecasts, estimate the conditional information retained by faithful reports. For recommendations, compare expected principal value under the recommended actions.

\paragraph{Own-label leakage $\varepsilon_i$.}
Construct paired evaluations in which the public decision $Z$, other evidence, and outside data are held fixed while adviser $i$'s report is varied, redacted, or replaced. Estimate a divergence between the resulting operational-label distribution and a leave-one-out or holdout reference. Total variation is the theorem's quantity; classifiers and calibrated density-ratio estimators provide practical lower bounds. Architectural certificates complement these statistical audits from the other side: by \cref{prop:dp}, a differential-privacy budget on the evaluator's report slot is a provable upper bound on $\varepsilon_i$.

\paragraph{Decision interest $D_i$.}
Estimate how much an adviser can gain outside the score by steering the decision: winning a contract, selecting an easier task, changing downstream exposure, or favoring an affiliate. This quantity determines how large the score margin must be even under perfect factorization.

\paragraph{Constitution sensitivity.}
For neighboring constitutions $C,C'$, define
\[
    \kappa(C,C')=\TV\big(\Law(Z\mid C),\Law(Z\mid C')\big).
\]
Paraphrase tests should distinguish benign semantic invariance from legitimate changes in institutional meaning.

\paragraph{Memo fidelity.}
Measure
\[
    m=\Prb[\chi(A)=Z]
\]
and audit whether every material sentence in $A$ is supported by provenance $P$. A high-quality explanation that cannot recover the settled object is an institutional inconsistency, not merely a summarization error.

\subsection{Suggested experiments}

Four experiments would directly test the paper's claims.

\begin{enumerate}[leftmargin=1.5em]
    \item \textbf{Scalar markets versus language advice.} Give LLM agents identical private-signal environments and compare scalar prediction markets with reports containing conditional forecasts, rationales, and evidence. Existing experiments provide a market baseline \parencite{galanis2026aggregation}.
    \item \textbf{Leakage curves.} Gradually reveal an adviser's report to its evaluator while holding the public decision fixed. Plot truthful payoff margin against estimated $\varepsilon_i$ and compare with the $2L_i\varepsilon_i$ envelope; \cref{prop:tightness} predicts that pandering-style manipulations against a biased evaluator can saturate it.
    \item \textbf{Mecha-nudge stress tests.} Preserve hardened factual content while varying presentation, ordering, and embedded instructions. A robust hardener should keep $e_i$ and the externalized label stable even when the unfactorized judge changes its decision \parencite{frey2026mecha}.
    \item \textbf{Hybrid-label ablations.} Add tests, audits, peer labels, and outside outcomes one at a time. Measure both incremental information margin and incremental leakage, and audit the independence hypothesis of \cref{lem:composition} rather than only per-component leakage. This tests the central tradeoff of \cref{sec:labels}.
\end{enumerate}

\section{The mechanism-level normal form: what remains conjectural}
\label{sec:conjecture}

\Cref{sec:normalform} proved a normal form at the level of settlement channels: every presentation of every mechanism carries a canonical ghost reference, intrinsic leakage is a well-defined invariant within a factor of two of the best achievable reference, and worst-case erosion identifies it (\cref{thm:ghost,prop:converse}). A strictly stronger claim remains open at the level of whole mechanisms. Call two mechanisms \emph{$\delta$-equivalent} if, for every report profile and signal realization, their induced joint distributions over the public record, transfers, and written resolution are within $\delta$ in total variation.

\begin{conjecture}[Approximate normal form for robust advice mechanisms; informal]
Consider a language advice mechanism whose truthful equilibrium is robust to hardening-preserving paraphrases and whose payment-relevant behavior is stable under changes to the written-resolution channel that preserve the public decision record. Suppose all intended report influence on outcomes can be represented through an explicit public record $Z$. Then for some small $\delta$ the mechanism has a $\delta$-equivalent implementation with:
\begin{enumerate}[leftmargin=1.5em]
    \item a constitution-dependent hardening map;
    \item an explicit public decision record $Z$;
    \item agent-specific labels that are conditionally influence-free from the paid agent's own hardened advice;
    \item payment kernels evaluated against those labels; and
    \item a provenance-linked writer downstream of $Z$.
\end{enumerate}
Any remaining strategic difference from the original implementation can be represented by an own-label leakage term $\varepsilon_i$ and an explicit decision-coupling term inside the externalized advice margin.
\end{conjecture}

If true, the conjecture would make settlement factorization more than good engineering practice. It would say that robust language advice mechanisms can be compiled into a causal normal form: intended influence flows through the public decision; evaluation influence is externalized or priced.

The channel-level results prove the tractable part of this converse and delimit the rest. Hardening invariance supplies the representation step for reports; the ghost reference supplies the influence-free label constructively, presentation by presentation; and the erosion identity makes the residual $\varepsilon_i$ behaviorally meaningful. What the conjecture adds, per \cref{rem:normalform}, is $\delta$-equivalence of the full joint distribution---including the written channel---under the mechanism's own \emph{fixed} kernel, together with a canonical choice among presentations. A full proof would need to solve that minimization over presentations, characterize when arbitrary model-mediated evaluators admit a sufficient externalized label, and rule out coalition or identity effects that prevent adviser-wise factorization.

\section{Conclusion}

Natural language mechanisms are most useful when reports do more than predict. They advise: they propose actions, explain tradeoffs, provide evidence, and change what a principal does. That is why the correct robustness requirement is not to remove the report from the decision. It is to prevent the report from covertly controlling the yardstick used to pay itself.

\textcite{dellapenna2024nlm} identifies a positive end-to-end regime for self-resolution. This paper describes the broader institutional architecture. Harden reports into official advice. Let the public decision use that advice. Externalize each adviser's payment label. Generate the explanation from the settled record and its provenance. Then treat evaluator leakage as a measurable mechanism parameter.

Recent work supplies increasingly sophisticated ways to create truthful scores for text, evidence, rich queries, unverifiable outcomes, and decision-coupled predictions. Settlement factorization is the layer that lets those mechanisms compose. If an upstream design creates margin $\gamma_i$, direct decision interests cost at most $D_i$, and the operational evaluator has leakage $\varepsilon_i$ and sensitivity $L_i$, the institution has a concrete robustness test:
\[
    \gamma_i>D_i+2L_i\varepsilon_i.
\]
The inequality is not a complete theory of advice markets. It is a useful boundary: reward the information that changes the decision, but do not let the advice write its own answer key. And the test is universal: by the ghost-reference theorem, every settlement channel carries a canonical influence-free reference, so $\varepsilon_i$ is defined---and by the erosion converse, measurable---for any mechanism whatsoever, not only for those designed with factorization in mind.

\FloatBarrier
\appendix

\section{Proofs}

\begin{proof}[Proof of \cref{prop:hardening} (hardening invariance)]
If $r_i,\widetilde r_i\in T_i^C(s_i)$, then
\[
    h_i^C(r_i)=h_i^C(\widetilde r_i)=\phi_i^C(s_i).
\]
Under exact factorization, every downstream computation depends on reports only through hardened advice. Replacing $r_i$ by $\widetilde r_i$ therefore leaves the hardened profile unchanged, and hence leaves $Z$, all label distributions, transfers, provenance, and the written resolution unchanged. The argument is pointwise in $r_{-i}$.
\end{proof}

\begin{proof}[Proof of \cref{lem:tv-expectation} (TV--expectation bound)]
Since $(\mu-\nu)(\calY)=0$, subtracting the constant $\tfrac{a+b}{2}$ leaves the difference of expectations unchanged:
\[
    \E_\mu[f]-\E_\nu[f]=\int\Bigl(f-\tfrac{a+b}{2}\Bigr)\,d(\mu-\nu).
\]
The integrand is bounded by $\tfrac{b-a}{2}$ in absolute value, so the integral is bounded by $\tfrac{b-a}{2}\,\|\mu-\nu\|$, where $\|\mu-\nu\|$ is the total variation norm of the signed measure $\mu-\nu$. By the Hahn decomposition, $\|\mu-\nu\|=2\sup_A|\mu(A)-\nu(A)|=2\TV(\mu,\nu)$, giving the claim.
\end{proof}

\begin{proof}[Proof of \cref{prop:exact} (exact factorization preserves truthful advice)]
Fix $s$, adviser $i$, a faithful report $r_i^T$, and deviation $r_i'$. The baseline $b_i(e_{-i},X,U)$ is the same under both reports. Under exact factorization, the expected utility difference between faithful advice and the deviation is
\begin{align*}
&\E\big[G_i(\phi_i^C(s_i),Z^T,Y_i^{\circ,T})
-G_i(h_i^C(r_i'),Z^D,Y_i^{\circ,D})\mid s\big]\\
&\qquad+\E\big[w_i(Z^T,\theta)-w_i(Z^D,\theta)\mid s\big]\\
&=m_i^\circ(s,r_i')-d_i(s,r_i')
=\Gamma_i^\circ(s,r_i').
\end{align*}
By hypothesis this is nonnegative for every deviation, so every faithful report is a best response after every realized signal profile. Strictness follows whenever hardened advice changes and the margin is strict.
\end{proof}

\begin{proof}[Proof of \cref{cor:proper} (proper scoring)]
Condition on adviser $i$'s signal $s_i$. The truthful hardened advice induces the conditional law $q_i(\phi_i^C(s_i))=\Law(Y_i^\circ\mid s_i)$. Because $S_i$ is strictly proper and, by \cref{ass:neutrality}, the law of $Y_i^\circ$ is invariant to adviser $i$'s report, this forecast uniquely maximizes expected score among all alternative induced forecasts. The baseline is report-independent and $w_i\equiv0$, so faithful advice is a Bayesian best response. Strictness holds when an evidence-changing deviation changes the forecast with positive probability.
\end{proof}

\begin{proof}[Proof of \cref{prop:info} (information margin)]
Condition on $(s_i,R_i)$. For deterministic garblings $R_i=g_i(s_i)$, conditioning on $(s_i,R_i)$ is conditioning on $s_i$; for stochastic garblings, the Markov condition $Y_i^\circ\ind R_i\mid s_i$ gives $\Law(Y_i^\circ\mid s_i,R_i)=\Law(Y_i^\circ\mid s_i)=p_{s_i}$, so in either case the label law given $(s_i,R_i)$ is $p_{s_i}$. Properness then gives
\begin{align*}
&\E\left[S_i(p_{s_i},Y_i^\circ)-S_i(p_{R_i},Y_i^\circ)\mid s_i,R_i\right]\\
&\qquad= D_{S_i}(p_{s_i},p_{R_i}).
\end{align*}
Taking expectations yields the first identity. Under the logarithmic score $S(p,y)=\log p(y)$,
\begin{align*}
\E\left[\log\frac{p(Y_i^\circ\mid s_i)}{p(Y_i^\circ\mid R_i)}\right]
=I(Y_i^\circ;s_i\mid R_i),
\end{align*}
which is the definition of conditional mutual information.
\end{proof}

\begin{proof}[Proof of \cref{cor:decision} (decision-value score)]
Condition on $s_i$. By hypothesis, $\alpha_i(\phi_i^C(s_i))$ maximizes the conditional expected value $\E[v(a,Y_i^\circ)\mid s_i]$ over every action reachable through an admissible report. By \cref{ass:neutrality} the label law is invariant to the report, and the adviser has no direct decision utility, so no deviation can improve expected payment.
\end{proof}

\begin{proof}[Proof of \cref{thm:leakage} (approximate settlement factorization)]
Fix $s$, adviser $i$, a faithful report $r_i^T$, and a deviation $r_i'$; write $r^T,r^D$ for the induced report profiles with the other advisers faithful, and $e_i^T=\phi_i^C(s_i)$, $e_i^D=h_i^C(r_i')$. For each $z$ in the support of the relevant scenario, define
\begin{align*}
\mu_T(z)&=\Law(Y_i\mid s,r^T,Z=z),
&\mu_T^\circ(z)&=\Law(Y_i^\circ\mid s,e_{-i},Z=z),\\
\mu_D(z)&=\Law(Y_i\mid s,r^D,Z=z),
&\mu_D^\circ(z)&=\Law(Y_i^\circ\mid s,e_{-i},Z=z).
\end{align*}
By $\varepsilon_i$-leakage (\cref{def:leaky}),
\[
    \TV(\mu_T(z),\mu_T^\circ(z))\le\varepsilon_i
    \qquad\text{and}\qquad
    \TV(\mu_D(z),\mu_D^\circ(z))\le\varepsilon_i
\]
for almost every $z$ in the support of $Z$ under $(s,r^T)$ and $(s,r^D)$ respectively. Applying \cref{lem:tv-expectation} pointwise in $z$ with $f=G_i(e_i^T,z,\cdot)\in[a_i,b_i]$ and then with $f=G_i(e_i^D,z,\cdot)$,
\begin{align*}
\E_{Y\sim\mu_T(z)}[G_i(e_i^T,z,Y)]
&\ge \E_{Y\sim\mu_T^\circ(z)}[G_i(e_i^T,z,Y)]-L_i\varepsilon_i,\\
\E_{Y\sim\mu_D(z)}[G_i(e_i^D,z,Y)]
&\le \E_{Y\sim\mu_D^\circ(z)}[G_i(e_i^D,z,Y)]+L_i\varepsilon_i.
\end{align*}
Average the first inequality over $Z^T$ and the second over $Z^D$ by the tower property. This is legitimate because the law of $Z$ in each scenario is identical in the operational and reference systems: by item 1 of \cref{def:factorization}, $Z$ depends on reports only through hardened advice, and leakage enters only through the label. The direct decision-utility terms $w_i(Z,\theta)$ and the report-independent baseline $b_i$ are likewise identical across the two systems. Subtracting the deviation side from the truthful side gives
\[
    \Gamma_i(s,r_i')\ge\Gamma_i^\circ(s,r_i')-2L_i\varepsilon_i.
\]
Taking the infimum over hardened-advice-changing deviations gives the strict signal--ex-post equilibrium condition.
\end{proof}

\begin{proof}[Proof of \cref{cor:interests} (decision interests plus leakage)]
For every $s,r_i'$, one has $m_i^\circ(s,r_i')\ge\underline m_i$ and $d_i(s,r_i')\le D_i$. Hence
\[
    \Gamma_i^\circ(s,r_i')=m_i^\circ(s,r_i')-d_i(s,r_i')
    \ge \underline m_i-D_i.
\]
The leakage theorem then gives
\[
    \Gamma_i(s,r_i')\ge \underline m_i-D_i-2L_i\varepsilon_i>0.
\]
\end{proof}

\begin{proof}[Proof of \cref{prop:brier} (binary own-label manipulation)]
Let truthful advice report $q=1/2$ and let its operational label equal $Y^\circ\sim\mathrm{Bernoulli}(1/2)$. Its expected Brier payment is
\[
    \E[-(1/2-Y^\circ)^2]=-1/4.
\]
Consider a deviation that reports $q'=1/2+\varepsilon$ and shifts the operational label to $Y'\sim\mathrm{Bernoulli}(1/2+\varepsilon)$. The total variation distance between the two Bernoulli laws is $\varepsilon$. The deviating expected score is
\[
    \E[-(1/2+\varepsilon-Y')^2]
    =-\left(\tfrac12+\varepsilon\right)\left(\tfrac12-\varepsilon\right)
    =-\tfrac14+\varepsilon^2.
\]
The deviation therefore improves expected payment by $\varepsilon^2$.
\end{proof}

\begin{proof}[Proof of \cref{prop:tightness} (exact tightness)]
The environment has a single adviser, a binary label space, a trivial public record, and $w\equiv 0$, so margins consist of the score term alone and the baseline cancels. The kernel $G$ takes values in $[0,1]$, so it is $1$-stable by \cref{lem:tv-expectation}; take $L=1$. Truthful hardened advice encodes the posterior $q=1/2$ and thus takes the branch $q\le 1/2$, earning $\E[1-Y]$ under whichever label law applies; the deviation $q'=1/2+\varepsilon$ takes the branch $q>1/2$, earning $\E[Y]$.

Under the reference label $Y^\circ\sim\mathrm{Bernoulli}(1/2)$, every report earns expected payment $1/2$, so $\Gamma^\circ(s,r')=0$ for every deviation, and faithful reporting is weakly optimal in the reference system, consistent with \cref{prop:exact}.

Part (a). The operational label deviates from the reference only when the report declares $q>1/2$, and then $\TV(\mathrm{Bernoulli}(1/2+\varepsilon),\mathrm{Bernoulli}(1/2))=\varepsilon$; otherwise the laws coincide. Hence $Y$ is $\varepsilon$-leaky. Truthful advice leaves the label untouched and earns $1/2$; the deviation earns $\E[Y']=1/2+\varepsilon$. Therefore $\Gamma(s,r')=1/2-(1/2+\varepsilon)=-\varepsilon=\Gamma^\circ(s,r')-L\varepsilon$.

Part (b). The operational label law is $\mathrm{Bernoulli}(1/2+\varepsilon)$ under every report, so every conditional comparison with the reference has total variation exactly $\varepsilon$, and $Y$ is $\varepsilon$-leaky. Truthful advice now earns $\E[1-Y]=1/2-\varepsilon$, while the deviation earns $\E[Y]=1/2+\varepsilon$. Therefore $\Gamma(s,r')=-2\varepsilon=\Gamma^\circ(s,r')-2L\varepsilon$, matching the bound of \cref{thm:leakage} with equality: the truthful side loses $L\varepsilon$ to the adverse bias and the deviation side gains $L\varepsilon$ by pandering to it.
\end{proof}

\begin{proof}[Proof of \cref{lem:composition} (leakage composition)]
Fix conditioning values $(s,r,Z=z)$ in the support and abbreviate the operational component laws $\mu_k=\Law(C_k\mid s,r,Z=z)$ and the reference component laws $\nu_k=\Law(C_k^\circ\mid s,e_{-i},Z=z)$; by componentwise leakage, $\TV(\mu_k,\nu_k)\le\varepsilon_k$, and by the independence hypotheses the joint operational law is $\bigotimes_k\mu_k$ and the joint reference law is $\bigotimes_k\nu_k$. Two standard facts complete the argument.

(i) \emph{Product telescoping:} $\TV(\bigotimes_k\mu_k,\bigotimes_k\nu_k)\le\sum_k\TV(\mu_k,\nu_k)$. Define the hybrid measures $h_j=\mu_1\otimes\cdots\otimes\mu_j\otimes\nu_{j+1}\otimes\cdots\otimes\nu_K$ for $j=0,\ldots,K$, so $h_0=\bigotimes_k\nu_k$ and $h_K=\bigotimes_k\mu_k$. Consecutive hybrids differ in a single factor, and for measures sharing a common factor $\rho$ one has $\TV(\alpha\otimes\rho,\beta\otimes\rho)=\TV(\alpha,\beta)$: take an optimal coupling of $(\alpha,\beta)$ and couple the $\rho$-coordinates identically, so the coupled pair disagrees exactly when the first coordinates do; the reverse inequality follows by restricting to rectangle events. The triangle inequality across the $K$ hybrid steps gives the sum.

(ii) \emph{Data processing:} for any measurable $h$, $\TV(h_*\mu,h_*\nu)\le\TV(\mu,\nu)$, since events of the form $\{h\in A\}$ are a sub-collection of the measurable events.

Combining, $\TV(\Law(Y_i\mid s,r,Z=z),\Law(Y_i^\circ\mid s,e_{-i},Z=z))\le\sum_k\varepsilon_k$ for almost every conditioning value, which is the claim.
\end{proof}

\begin{proof}[Proof of \cref{cor:stakes} (stakes)]
If $G_i$ is $L_i$-stable, then for any $\mu,\nu$,
$|\E_\mu[cG_i]-\E_\nu[cG_i]|=c\,|\E_\mu[G_i]-\E_\nu[G_i]|\le cL_i\TV(\mu,\nu)$,
so $G_i^c$ is $cL_i$-stable. Every score term in \cref{eq:score-margin} scales by $c$, so $m_i^{\circ,c}=c\,m_i^\circ$, while \cref{eq:decision-gain} does not involve the transfer and is unchanged. Applying \cref{thm:leakage} to the scaled kernel,
\[
\Gamma_i^c(s,r_i')\ \ge\ c\,m_i^\circ(s,r_i')-d_i(s,r_i')-2cL_i\varepsilon_i
\ \ge\ c\,(\underline m_i-2L_i\varepsilon_i)-D_i,
\]
which is strictly positive under the stated condition; part (a) follows by solving for $c$. For part (b), scale the environment of \cref{prop:tightness}(b) by $c$: the kernel $cG$ takes values in $[0,c]$, truthful advice earns $c(1/2-\varepsilon)$ and the pandering deviation earns $c(1/2+\varepsilon)$, so the deviation profit is $2c\varepsilon=2cL\varepsilon$, strictly increasing in $c$.
\end{proof}

\begin{proof}[Proof of \cref{prop:audit} (audited settlement)]
Fix a conditioning scenario of \cref{def:leaky} and write $\mu=\Law(Y_i\mid s,r,Z=z)$ and $\mu^\circ=\Law(Y_i^\circ\mid s,e_{-i},Z=z)$, with $\TV(\mu,\mu^\circ)\le\varepsilon_i$. Because the audit lottery is independent of the scenario, the mixture label has conditional law $p_i\mu^\circ+(1-p_i)\mu$ in that same scenario. For every measurable $A$,
\[
\bigl|\bigl(p_i\mu^\circ+(1-p_i)\mu\bigr)(A)-\mu^\circ(A)\bigr|
=(1-p_i)\,\bigl|\mu(A)-\mu^\circ(A)\bigr|,
\]
so $\TV\bigl(p_i\mu^\circ+(1-p_i)\mu,\ \mu^\circ\bigr)=(1-p_i)\TV(\mu,\mu^\circ)\le(1-p_i)\varepsilon_i$, with equality when the underlying leakage is exact. The mixture is therefore $(1-p_i)\varepsilon_i$-leaky, and \cref{thm:leakage} gives $\Gamma_i\ge\Gamma_i^\circ-2L_i(1-p_i)\varepsilon_i$; the threshold $p_i^\ast$ follows by solving $\gamma_i>D_i+2L_i(1-p_i)\varepsilon_i$ for $p_i$ and truncating at zero. Expected settlement cost is $(1-p_i)k_i+p_ik_i^\circ=k_i+p_i(k_i^\circ-k_i)$, increasing in $p_i$ since $k_i^\circ>k_i$, hence minimized at the threshold.
\end{proof}

\begin{proof}[Proof of \cref{prop:dp} (differential privacy certifies leakage)]
Fix a conditioning scenario and write $\mu_{r_i}=\Law(Y_i\mid s,(r_i,r_{-i}),Z=z)$. Taking $\widetilde r_i=\bot_i$, the DP inequality gives, for every measurable $A$,
\[
\mu_{r_i}(A)-\mu_{\bot_i}(A)
\le \mu_{r_i}(A)-e^{-\delta_i}\mu_{r_i}(A)
=(1-e^{-\delta_i})\,\mu_{r_i}(A)\le 1-e^{-\delta_i},
\]
and symmetrically with the roles of $r_i$ and $\bot_i$ exchanged. Hence $\TV(\mu_{r_i},\mu_{\bot_i})\le 1-e^{-\delta_i}\le\delta_i$. Since the reference pipeline is by definition the operational pipeline run on $\bot_i$, this is exactly the bound of \cref{def:leaky}.

(a) For measurable $h$ and report-independent randomness $\omega$, events of the form $\{h(Y_i,\omega)\in A\}$ have probabilities that are averages over $\omega$ of probabilities of events $\{Y_i\in A_\omega\}$; each satisfies the likelihood-ratio bound, hence so does the average.

(b) We show the joint component vector satisfies the DP inequality with parameter $\sum_k\delta_{i,k}$; the claim then follows from the first display and part (a) applied to $h$. Work with densities with respect to a common dominating measure, conditioning throughout on the scenario. Let $f_{r}(c_{1:K})$ denote the joint density of $(C_1,\ldots,C_K)$ under report $r$. By the chain rule and the stagewise hypothesis,
\[
\frac{f_{r_i}(c_{1:K})}{f_{\widetilde r_i}(c_{1:K})}
=\prod_{k=1}^{K}
\frac{f_{r_i}(c_k\mid c_{1:k-1})}{f_{\widetilde r_i}(c_k\mid c_{1:k-1})}
\le \prod_{k=1}^{K} e^{\delta_{i,k}}
= e^{\sum_k\delta_{i,k}}
\]
for almost every $c_{1:K}$, where each factor is bounded because stage $k$ is $\delta_{i,k}$-DP for \emph{every} realization $c_{1:k-1}$ of the earlier components---this is what tolerates adaptivity. Integrating the density bound over any measurable set yields the DP inequality for the joint vector, completing the proof.
\end{proof}

\FloatBarrier
\printbibliography

\end{document}